\documentclass[11pt]{amsart}

\usepackage{amsmath, latexsym, amsfonts, amssymb, amsthm}
\usepackage{graphicx, color,hyperref,epsfig,caption,subfig}
\usepackage{pstricks}
\usepackage{mathrsfs}
\usepackage{verbatim}

\hypersetup{
    linktoc=page,
    linkcolor=red,          
    citecolor=blue,        
    filecolor=blue,      
    urlcolor=cyan,
    colorlinks=true           
}

\numberwithin{equation}{section}

\newtheorem{theorem}{Theorem}[section]
\newtheorem{lemma}[theorem]{Lemma}
\newtheorem{proposition}[theorem]{Proposition}

\newtheorem{remark}[theorem]{Remark}
\newtheorem{definition}[theorem]{Definition}

\theoremstyle{definition}

\newcommand{\C}{\mathbb{C}}
\newcommand{\R}{\mathbb{R}}

\newcommand{\E}{\mathbb{E}}
\newcommand{\D}{\mathbb{D}}

\newcommand{\hf}{\frac{_1}{^2}}
\newcommand{\thf}{\tfrac{1}{2}}
\def\eps{\epsilon}

\newcommand{\bx}{\textbf{x}}
\newcommand{\by}{\textbf{y}}
\newcommand{\bl}{\textbf{l}}

\newcommand{\caB}{{\mathcal B}}
\newcommand{\caC}{{\mathcal C}}
\newcommand{\caD}{{\mathcal D}}

\newcommand{\caF}{{\mathcal F}}

\newcommand{\caH}{{\mathcal H}}

\newcommand{\caM}{{\mathcal M}}

\newcommand{\caO}{{\mathcal O}}

\newcommand{\caT}{{\mathcal T}}

\newcommand{\caV}{{\mathcal V}}

\newcommand{\opn}{\operatorname}

\newcommand{\rs}{{\hat \C}}

\newcommand{\de}{\tfrac{d}{d\eps} \big|_{\eps=0}}

\begin{document}


\title[Stress-Energy in Liouville Conformal Field Theory]{Stress-Energy in Liouville Conformal Field Theory}

\author[Antti Kupiainen]{Antti Kupiainen$^{1}$}
\address{University of Helsinki, Department of Mathematics and Statistics,
         P.O. Box 68 , FIN-00014 University of Helsinki, Finland}
\email{antti.kupiainen@helsinki.fi}

\author[Joona Oikarinen]{Joona Oikarinen$^{1}$}
\address{University of Helsinki, Department of Mathematics and Statistics,
         P.O. Box 68 , FIN-00014 University of Helsinki, Finland}
\email{ joona.oikarinen@helsinki.fi}

\footnotetext[1]{Supported by the Academy of Finland and ERC Advanced Grant 741487}
\begin{abstract}
We construct the stress-energy tensor correlation functions
 in probabilistic Liouville Conformal Field Theory (LCFT) on the two-dimensional sphere $\mathbb{S}^2$ by studying the variation of the LCFT correlation functions with respect to a smooth  Riemannian metric on  $\mathbb{S}^2$. In particular we derive conformal Ward identities for these correlation functions. This forms the basis for the   construction of a  representation of the Virasoro algebra on the canonical Hilbert space of the LCFT.
  In \cite{ward} the conformal Ward identities were derived for one and two stress-energy tensor insertions using a different definition of the stress-energy tensor and Gaussian integration by parts. By defining the stress-energy correlation functions as functional derivatives of the LCFT correlation functions and using the smoothness of the LCFT correlation functions proven in \cite{Oik} allows us to control an arbitrary number of  stress-energy tensor insertions needed for representation theory.
\end{abstract}

\maketitle

\tableofcontents

\section{Introduction and Main result}


\subsection{Local conformal symmetry} Two-dimensional conformal field theory (CFT) is characterized by local conformal symmetry, an infinite dimensional symmetry that strongly constrains the theory. 
A  formulation of this symmetry can be summarised as follows \cite{Gaw}. The basic data of CFT are {\it correlation functions} 
\begin{align}
\langle \prod_{i=1}^NV_{\alpha_{i}}(x_{i})\rangle_{\Sigma,g}
\label{correlat}
\end{align}
of {\it primary fields} $V_{\alpha}(z)$ defined on  a compact  two-dimensional surface $\Sigma$ equipped with a smooth Riemannian metric $g$. In a probabilistic formulation of CFT the angular bracket $\langle\cdot\rangle_{\Sigma,g}$ is an { expectation} in a suitable positive (not necessarily probability) measure and the primary fields often are (distribution valued) random fields. 

The  local conformal symmetry arises from the transformation properties of the correlation functions under the action of the groups of smooth diffeomorphisms and local Weyl transformations of the metric. The former acts by pullback on the metric  $g \mapsto \psi^\ast g$ and  the latter acts by a local scale transformation $g \mapsto e^\varphi g$ with $\varphi\in C^\infty(\Sigma,\R)$. In axiomatic CFT one postulates\footnote{for scalar fields $V_{\alpha}$}
\begin{align}
\langle\prod_{i=1}^N V_{\alpha_{i}}(x_{i})\rangle_{\Sigma,\psi^{\ast}g} &= \langle \prod_{i=1}^NV_{\alpha_{i}}(\psi(x_{i}))\rangle_{\Sigma,g}\label{diffeo}\\
\langle \prod_{i=1}^N V_{\alpha_{i}}(x_{i})\rangle_{\Sigma,e^{\varphi}g}&=e^{cA(\varphi,g)}\prod_{i=1}^N e^{-\Delta_{\alpha_{i}}\varphi(x_{i})}\langle \prod_{i=1}^N V_{\alpha_{i}}(x_{i})\rangle_{\Sigma,g}\label{weyl}
\end{align}
where the {\it conformal anomaly} is given by
\begin{align}
A(\varphi,g)&=\tfrac{1}{96\pi}\int( |\nabla_{g}\varphi|^{2}+2R_{g}\varphi)dv_{g}
\label{anomaly}
\end{align}
and the constant $c$ is the central charge of the CFT, which in our case will belong to the interval $(25,\infty)$. The number $\Delta_{\alpha}$ is called the conformal weight of the field $V_{\alpha}$. We denoted by $v_g$ the Riemannian volume measure and by $R_{g}$ the curvature scalar (see Appendix).

The {\it stress-energy tensor} field $T_{\mu\nu}(z)$ is a symmetric 2 by 2 complex matrix valued field defined indirectly through the formal variation of the (inverse of the) metric at a point $z \in \Sigma$ in the correlation function \eqref{correlat} (the precise definition can be found in Section \ref{se_section}): 
\begin{align}\label{Tcorre}
\langle \prod_{i=1}^n T_{\mu_i \nu_i}(z_i) \prod_{j=1}^N V_{\alpha_{j}}(x_{j}) \rangle_{\Sigma,g} &:= (4 \pi)^n \prod_{i=1}^n \frac{\delta }{\delta g^{\mu_i \nu_i}(z_i)}  \langle \prod_{j=1}^N V_{\alpha_{j}}(x_{j}) \rangle_{\Sigma,g}  \,,
\end{align}
where $g^{\mu_i \nu_i}$ denotes a component of the inverse of the metric $g$. The functions on the right-hand side of \eqref{Tcorre} turn out to be analytic or anti-analytic in the variables $z_i$ in the complex coordinates \eqref{ccoordinates} as long as $z_i \neq z_j$, $x_i \neq x_j$ for $i\neq j$ and $z_i \neq x_j$ for all $i$ and $j$. These functions diverge when two variables merge but for certain choices of the indices $\mu_i$ and $\nu_i$ the functions turn out to be meromorphic with poles described by the conformal Ward identities. Let us specialize to the case of the sphere, $\Sigma={\mathbb S}^2$. Then every smooth metric $g$ can be obtained from a given one $\hat g$ by the action of diffeomorphisms and Weyl transformations:
\begin{align*}
g=e^\varphi\psi^\ast \hat g.
\end{align*}
This fact together with the symmetries \eqref{diffeo} and \eqref{weyl} yields the tools for defining and computing the functional derivatives on the right-hand side of \eqref{Tcorre}. 
The result is a recursive formula, 
the {\it conformal Ward identity}, that allows to express \eqref{Tcorre} in terms of derivatives (in the $x_i$'s) of \eqref{correlat}. 

Ward identities take an especially simple form in complex coordinates. Recall that a Riemannian metric determines a {\it complex structure} on $\Sigma$: a system of local coordinates where the metric takes the form
\begin{align*}
\hat g=\thf e^\sigma(dz\otimes d\bar z+d\bar  z\otimes dz).
\end{align*}
In such coordinates consider the $zz$-component of the stress-energy, 
\begin{align*}
T_{zz}(z)=\tfrac{1}{4}(T_{11}-T_{22}-2iT_{12})
\end{align*}
where $\{T_{ij}\}_{i,j=1}^2$ are the components of $T$ in the Euclidean coordinates of the plane. We define
\begin{align}
T(z)=T_{zz}(z)+\frac{c}{12}t(z)
\label{newT}
\end{align}
where 
\begin{align}\label{tcorre}
t(z) &=  \partial_z^2 \sigma(z)-\hf (\partial_z \sigma(z))^2.
\end{align}
Then for distinct points $\{z_i,x_j\}$ the Ward identity in the case $\Sigma = \mathbb{S}^2$ reads 
\begin{align}\nonumber
\langle \prod_{k=1}^n T(z_k) \prod_{i=1}^N V_{\alpha_i}(x_i) &\rangle_{\mathbb{S}^2,g} =\hf 
\sum_{j=2}^n  \frac{c}{(z_1-z_j)^4} 
 \langle \prod_{k\neq 1,j}^n T(z_k) \prod_{i=1}^N V_{\alpha_i}(x_i) \rangle_{\mathbb{S}^2,g} \\
& \quad + \sum_{j=2}^n \left( \frac{2}{(z_1-z_j)^2} + \frac{\partial_{z_j}}{z_1-z_j} \right) \langle \prod_{k=2}^n T(z_k) \prod_{i=1}^N V_{\alpha_i}(x_i) \rangle_{\mathbb{S}^2,g} \nonumber \\
& \quad + \sum_{j=1}^N \left(  \frac{\Delta_{\alpha_j}}{(z_1-x_j)^2}+ \frac{\Delta_{\alpha_j}\partial_z\sigma(x_j)}{z_1-x_j}+\frac{1}{z_1-x_j}\partial_{x_j}\right) \langle \prod_{k=2}^n T(z_k) \prod_{i=1}^N V_{\alpha_i}(x_i) \rangle_{\mathbb{S}^2,g}.\label{wardid}
\end{align}
The definition \eqref{newT} was already introduced in \cite{EgOo}, and it is natural in the sense that it makes $T(z)$ meromorphic also in the regions where there are curvature, which is essential for the Virasoro algebra discussed in Section \ref{prospects}. Iterating this identity, the left-hand side of \eqref{Tcorre} in the case $(\mu_i,\nu_i)=(z,z)$ will be expressed in terms of the functions $\eqref{correlat}$ and their derivatives in the $x_i$'s. 
A corollary of this identity is that the functions  \eqref{Tcorre}  are holomorphic in the variables $\{z_i\}$ in the region $\{z_i \neq z_j : i \neq j\} \cap \{z_i \neq x_j : \forall \, i\,, j\}$.

The conformal Ward identities have been studied before with theoretical physics level of rigour. For a flat background metric the conformal Ward identities were initially derived in \cite{BPZ}. 
For a general metric and surface  the identities were derived in \cite{EgOo}, where also a term dealing with variation of the moduli of the surface appears. This term originates from the fact that compact Riemann surfaces with positive genus have non-trivial moduli spaces, so variation of the metric can also vary the conformal class of the surface.

\subsection{Path Integrals and Liouville Conformal Field Theory }

In Constructive Quantum Field Theory one attempts to construct the expectation 
as a path integral
\begin{align}\label{path_integral}
\langle F \rangle_{\Sigma,g} &:=  \int F(\phi) e^{-S(\phi,g)} \, D \phi\,,
\end{align}
over some space of fields $\phi:\Sigma\to \R$ (in the scalar case). The symmetries \eqref{diffeo} and \eqref{weyl} should then arise from the corresponding symmetries of the action functional $S$ with the anomaly \eqref{anomaly} arising from the singular nature of the integral in \eqref{path_integral}. A case where this program can be carried out is the
Liouville Conformal Field Theory (LCFT hereafter) which was introduced in 1981 by Polyakov \cite{Pol} in the context of developing a path integral theory for two-dimensional Riemannian metrics. 

Liouville field theory is described by the \emph{Liouville action functional}, which for $\gamma \in (0,2)$ and $\mu>0$ is given by
\begin{align}\label{action}
S_L(\phi,g) &= \tfrac{1}{\pi} \int_\Sigma \left( |\nabla_g \phi(z)|^2 + \tfrac{Q}{4} R_g(z) \phi(z) + \pi \mu e^{\gamma \phi(z)}  \right)  dv_g(z)\,.
\end{align}
The term $Q$ is  given by
$$Q = \tfrac{2}{\gamma} + \tfrac{\gamma}{2}.$$ 
The primary fields for LCFT are the \emph{vertex operators}
\begin{align*}
V_\alpha(z)=e^{\alpha\phi(z)}
\end{align*}
where $\alpha\in\C$. Their conformal weights are given by 
$$\Delta_\alpha = \tfrac{\alpha}{2}(Q-\tfrac{\alpha}{2}).$$

A rigorous construction of the path integral, and in particular the correlation functions of the vertex operators, was given 
 in \cite{DKRV} and will be recalled in Section \ref{covariant_lcft_section} in the present setup. 
 In  \cite{ward}  the conformal Ward identities \eqref{wardid} were derived in the case of one or two $T$-insertions ($n=1,2$). 
   Instead of deriving the conformal Ward identities by varying the background metric, the authors of \cite{ward} defined (the $zz$-component of) the stress-energy tensor directly as the field
\begin{align}\label{fieldem}
T_{zz}(x) &= Q \partial_x^2 \phi(x) - (\partial_x \phi(x))^2
\end{align}
and computed the correlation functions  \eqref{Tcorre} for $n=1,2$
for a specific metric by Gaussian integration by parts. Generalizing this approach to arbitrary $n$ was obstructed by a lack of proof of smoothness of the correlation functions  \eqref{correlat} (which was later proven \cite{Oik}) and the difficulty of simplifying the expressions coming from the integration by parts. It is however necessary to have \eqref{wardid} for arbitrary $n$ in order to construct the representation of the Virasoro algebra for LCFT. This is the motivation and the objective of the present paper.  Its main technical input is the recent proof of smoothness of the LCFT correlation functions by the second author \cite{Oik}. 

\subsection{Main result} Our main  result is a proof of the conformal Ward identities for arbitrarily many $T$-insertions for arbitrary metrics on the sphere by varying the background metric. The precise result is formulated in Propositions \ref{weyl_anomaly}, \ref{smoothness} and \ref{ward_prop}. 

\begin{theorem}\label{main_results}
Let $(x_1,\hdots,x_N) \in (\mathbb{S}^2)^N$, with $N \geq 3$, be a tuple of non-coinciding points on the two-dimensional sphere and assume that the real numbers $\alpha_1,\dots,\alpha_N$ satisfy the Seiberg bounds. The LCFT correlation functions \eqref{correlat} are smooth functions with respect to the Riemannian metric $g$, and they satisfy the diffeomorphism and Weyl symmetries \eqref{diffeo}, \eqref{weyl}. The derivatives \eqref{Tcorre} exist and are smooth in $z_i,x_j$ in the region of non-coinciding points. The correlations for the
field $T$  defined by \eqref{newT} satisfy the Ward identities \eqref{wardid}.

\end{theorem}




The content of the article is as follows. In Section \ref{covariant_lcft_section} we recall the definition of the correlation functions \eqref{correlat} and formulate and prove the diffeomorphism and Weyl covariance of LCFT on a compact surface $\Sigma$. In Section \ref{ward_section} we prove the Ward identities \eqref{wardid} for $\Sigma = \mathbb{S}^2$ and in Section 
 \ref{prospects} we discuss future work 
 on the construction of the Virasoro representation of LCFT. The appendix collects the elementary definitions and notations from Riemannian geometry used in the paper.
 
\vskip 3mm

\noindent{\it Acknowledgements.} We thank Remi Rhodes, Vincent Vargas and Yichao Huang for helpful discussions. The authors would like to thank the Isaac Newton Institute for Mathematical Sciences, Cambridge, for support and hospitality during the programme Scaling limits, rough paths, quantum field theory where work on this paper was undertaken. We also thank the anonymous referees for careful reading of this work and for providing numerous valuable comments on it. The work is supported by  the Academy of Finland and ERC Advanced Grant 741487. The work of J.O. is also supported by DOMAST.

\section{Covariant formulation of LCFT}\label{covariant_lcft_section}

In this section we recall the construction of LCFT correlation functions given in \cite{DKRV} and extend it to include the diffeomorphism covariance \eqref{diffeo}. Similar discussion can be found in \cite{GRV16}, Sections 3 and 4, where the authors work on compact Riemann surfaces with genus $2$ or higher, but the cases of the sphere and the torus work almost the same way. The main mathematical objects appearing in the construction are the \emph{Gaussian Free Field} and \emph{Gaussian Multiplicative Chaos} which we need to define in a covariant way. The appendix collects the elementary definitions and notations from Riemannian geometry used in this section.

\subsection{Gaussian Free Field}

Let $(\Sigma,g)$ be a two-dimensional smooth compact Riemannian manifold and $\Delta_g$ be the Laplace--Beltrami operator. It is well-known that $\Delta_g$ is a positive self-adjoint operator on $L^2(\Sigma,dv_g)$. The set of of orthonormal eigenfunctions $e_{g,n}$, $n=0,1,\dots$,
$$-\Delta_g e_{g,n} = \lambda_{g,n} e_{g,n},$$
is complete in the sense that the $L^2$-closure of $\opn{span}(e_{g,n})_{n =0}^\infty$ is the whole space $L^2(\Sigma,dv_g)$. It holds that $\lambda_{g,n} >0$ for $n>0$  and $\lambda_{g,0} =0$ with $e_{g,0}$ the constant function.

The Gaussian Free Field (GFF) $X_g$ on the Riemannian surface $(\Sigma,g)$ is defined as the random generalised function
\begin{align*}
X_g &= \sqrt{2 \pi } \sum_{n=1}^\infty a_n \frac{e_{g,n}}{\sqrt{\lambda_{g,n}}}\,,
\end{align*}
where $a_n$ are independent and identically distributed standard Gaussians. The series converges in the negative order Sobolev space $H^{-s}(\Sigma,dv_g)$ for any $s > 0$ (see e.g. Section 4.2 of \cite{Dub}). 
The covariance  of $X_g$ has an integral kernel 
\begin{align*}
\E [ (X_g,f)_g (X_g,h)_g ] &= \int G_g(x,y) f(x) h(y) \, dv_g(x) \, dv_g(y)\,,
\end{align*}
where $v_g$ is the volume measure of $g$,
\begin{align*}
G_g(x,y) &= \sum_{n=1}^\infty \frac{e_{g,n}(x) e_{g,n}(y)}{\lambda_{g,n}},
\end{align*}
and $(X_g,f)_g$ denotes the dual bracket, so formally $(X_g,f)_g=\int_\Sigma X_g f \, dv_g$. This justifies the notation
$$\E X_g(x) X_g(y) := G_g(x,y),$$
even though $X_g$ is almost surely not a function. An application of the Plancherel theorem gives
\begin{align}\label{deltag}
-\tfrac{1}{2\pi}\Delta_g G_g(x,y)
=\tfrac{1}{\sqrt{\det g(x)}}\delta(x-y)-\tfrac{1}{v_{g}(\Sigma)},
\end{align}
that is, $G_g(x,y)$ is the Green function of $-\Delta_g$ having zero average on $\Sigma$:
\begin{align*}
\int_\Sigma G_g(x,y)dv_{g}(x)=0=\int_\Sigma G_g(x,y)dv_{g}(y).
\end{align*}
Define the average in the metric $g$ as
\begin{align*}
m_{g}(f):=
\tfrac{1}{v_{g}(\Sigma)}\int_\Sigma f(z) dv_{g}(z).
\end{align*}
Then the GFF satisfies the following covariance under diffeomorphisms and Weyl transformations:
\begin{proposition}\label{gffinva}
(a) Let $\psi\in{\rm Diff}(\Sigma)$. Then
\begin{align}
X_g\circ\psi\overset{law}{=}X_{\psi^{\ast}g}.
\label{gff1a}
\end{align}
(b) Let $g$ and $g'$ be conformally equivalent, that is, $g' = e^\varphi g$ with $\varphi \in C^\infty(\Sigma,\R)$. Then
\begin{align}\label{change11}
X_{g'}\stackrel{law}{=}X_{g}-m_{g'}(X_g).
\end{align}
\end{proposition}

\proof
(a) Follows from covariance of the  Laplacian:
\begin{align*}
\psi^\ast\Delta_g(\psi^{-1})^\ast=\Delta_{\psi^\ast g}
\end{align*}
where $\psi^\ast$ acts on functions by $\psi^\ast f=f\circ\psi$. Hence
$\psi^\ast e_{n,g}=e_{n,\psi^\ast g}
$ 
from which \eqref{gff1a} follows.

(b) We have $g'=e^{\varphi}g$ for some $\varphi \in C^\infty(\Sigma,\R)$. Since $m_{g'}(X_g)=\frac{1}{v_{g'}(\Sigma)}(X_g,e^{\varphi})_{g}$ the field $X=X_{g}-m_{g'}(X_g)$ has covariance
\begin{align*}
\E X(z)X(z')&=G_g(z,z')-\tfrac{1}{v_{g'}(\Sigma)} \Big(\int G_g(z,x)dv_{g'}(x)+\int G_g(x,z')dv_{g'}(x) \Big)\\& \quad +\tfrac{1}{v_{g'}(\Sigma)^{2}}\int G_g(x,y)dv_{g'}(x)dv_{g'}(y).
\end{align*}
Since $\Delta_{e^{\varphi}g}=e^{-\varphi}\Delta_g $ we get from \eqref{deltag}
\begin{align*}
-\tfrac{1}{2\pi}\Delta_{g'} G_g(z,z')
=\tfrac{1}{\sqrt{\det g'(z)}}\delta(z-z')-\tfrac{1}{v_{g}(\Sigma)}e^{-\varphi}
\end{align*}
and thus
\begin{align*}
-\tfrac{1}{2\pi}\Delta_{g'} \E X(z)X(z')=\tfrac{1}{\sqrt{\det g'(z)}}\delta(z-z')-\tfrac{1}{v_{g'}(\Sigma)}.
\end{align*}
This implies that $\E X(z)X(z') = G_{g'}(z,z')$, meaning that  $X$ has the same covariance as $X_{g'}$. Since the fields are Gaussian, the equality in distribution follows. \qed

\vskip 2mm

Choose now a local conformal coordinate $z$ on $U\subset\Sigma$ so that the metric is
\begin{align}\label{conformalcoord}
g(z)=\tfrac{1}{2}e^{\sigma(z)}(dz\otimes d\bar z+d\bar z\otimes dz).
\end{align}
Then the equation \eqref{deltag} becomes 
\begin{align}\label{deltag2}
-\tfrac{1}{2\pi}\Delta G_g(x,y)
=\delta(x-y)-\tfrac{1}{v_{g}(\Sigma)}e^\sigma
\end{align}
where $\Delta$ is the standard Laplacian. Hence,
for $z,z'\in U$ we have
\begin{align}\label{gff_correlation}
G_g(z,z') &= \ln \frac{1}{|z-z'|} +h_g(z,z')
\end{align}
where $h_g$ is a smooth function ensuring that $G_g$ has zero $v_g$-mean over $\Sigma$ both in $z$ and $z'$. For later purpose we note that if $\psi$ is conformal on $U$ then
\begin{align*}
G_g(\psi(z),\psi(z'))=G_{\psi^*g}(z,z')=\ln \frac{1}{|z-z'|} +h_{\psi^*g}(z,z')
\end{align*}
so that 
\begin{align}\label{h_conformal}
h_g(\psi(x),\psi(x)) &= h_{\psi^*g}(x,x) + \ln |\psi'(x)|\,.
\end{align}
In particular for $\Sigma=\mathbb{S}^2$ we take $U=\C$ and have
\begin{align}\label{covg1}
h_{g}(z,z') 
=\tfrac{1}{v_{g}(\C)^2}\int_{\C}\int_{\C} \ln\frac{|z-u||z'-u|}{|u-v|}d v_g(u)dv_g(v).
\end{align}

\subsection{Gaussian Multiplicative Chaos}

Next, we want to define the measure $e^{\gamma X_g} dv_g$.  
We regularise the GFF by setting
\begin{align}\label{circle_average}
X_{g,N}(z) &= \sqrt{2 \pi} \sum_{n=1}^N \frac{a_n}{\sqrt{\lambda_n}} e_{g,n}(z)\,.
\end{align}
Then
 \begin{align*}
\E e^{\gamma X_{g,N}(x)}=e^{\frac{\gamma^{2}}{2}\E X_{{ g},N}(x)^{2}}\to\infty
\end{align*}
 as $N\to\infty$. Hence it is natural to consider
 the measure 
\begin{align*}
dm_{\gamma,g,N}(x)=e^{\gamma X_{g,N}(x)-\frac{\gamma^{2}}{2}\E X_{{ g},N}(x)^{2}}dv_g(x).
\end{align*}
If $\psi\in {\rm Diff(\Sigma)}$  then the properties $\psi^* e_{g,n} = e_{ \psi^*g,n}$ and $\lambda_{\psi^* g,n} = \lambda_{g,n}$ imply 
$$\psi^* X_{g,N} = X_{\psi^* g,N}.$$
Hence
\begin{align}
\psi^{\ast} m_{\gamma,g,N}=m_{\gamma, \psi^{\ast}g,N}.
\label{actio00000}
\end{align}
Let $\Sigma_{N}$ be the sigma-algebra generated by $a_{1},\dots,a_{N}$. Then, for $M<N$ and any continuous $f:\Sigma\to\R$
\begin{align*}
\E(\int_\Sigma f d m_{\gamma,g,N}|\Sigma_{M})=\int_\Sigma f d m_{\gamma,g,M},
\end{align*}
that is, the integrals against continuous functions are martingales. This leads to the  almost sure existence of the weak limit 
\begin{align*}
m_{\gamma,g}=\lim_{N\to\infty}m_{\gamma, g,N}
\end{align*}
and the limiting measure is non-trivial for $\gamma < 2$ (see e.g. \cite{Ber}) which is the origin of the parameter range $\gamma \in (0,2)$ that we mentioned in the introduction. The critical value $\gamma=2$ also leads to a non-trivial measure (see e.g. \cite{DRSV}), but requires a different renormalisation so we choose not to include it. The limiting measure is an instance of the {\it Gaussian Multiplicative Chaos} and it inherits the property \eqref{actio00000}
\begin{align}
\psi^{\ast} m_{\gamma,g}=m_{\gamma, \psi^{\ast}g}.
\label{actio000000}
\end{align}

\subsection{Weyl invariance}

To have the Weyl transformation law for LCFT we need to modify the chaos measure a bit. 
Fix conformal coordinates so that \eqref{conformalcoord} holds and consider the {\it circle average regularization} of $X_{g}$ given by
\begin{equation}\label{circfreg}X_{g,\epsilon}(z):=\tfrac{1}{2\pi}\int_0^{2\pi}X_g(z+\epsilon e^{i\theta})\,d\theta,
\end{equation}
see e.g. Lemma 3.2. in \cite{GRV16} for the precise definition of the circle average.
From \eqref{covg1} 
 and 
$$\int_0^{2\pi}\int_0^{2\pi}\ln |e^{i\theta}-e^{i\theta'}|\,d\theta d\theta'=0$$
 we deduce  
 \begin{equation}\label{circlegreen}
\E X_{g,\epsilon} (z)^2=\ln \epsilon^{-1}+h_g(z,z)+o(1),
\end{equation}
where $o(1)$ stands for terms that vanish as $\epsilon \to 0$.
It is known that the limit
\begin{align*}
\lim_{\epsilon\to 0}\epsilon^{\frac{_{\gamma^2}}{^2}}  e^{\gamma X_{{{ g}},\epsilon}(z)}
d^{2}z
\end{align*}
exists in the sense of weak convergence in probability (see e.g. \cite{Ber}). By uniqueness of the Gaussian Multiplicative Chaos measure (see \cite{Ber}), we have the equality
\begin{align*}
m_{\gamma,g}\overset{law}{=}e^{\sigma(z)}e^{-\frac{\gamma^2}{2}h_{{g}}(z,z)}\lim_{\epsilon\to 0}\epsilon^{\frac{_{\gamma^2}}{^2}}  e^{\gamma X_{{{ g}},\epsilon}(z)}
d^{2}z,
\end{align*}
where $d^2z$ denotes the two-dimensional Lebesgue measure. Let
\begin{align}\label{rhogamma}
 \rho_{\gamma,g}(z):=e^{\frac{_{ \gamma^2}}{^4}\sigma(z)+\frac{\gamma^2}{2}h_{g}(z,z)}%
\end{align}
and define the measure
\begin{align}
M_{\gamma,g}:=
\rho_{\gamma,g}m_{\gamma,g}.
\label{Mdefi}
\end{align}
This definition ensures that we get the transformation law under $g \mapsto e^\varphi g$ in Proposition \ref{diffeandweyl}, which will later turn out to be the correct one for the Liouville expectation in Section \ref{liouville_expectation_section} in the sense that it leads to the property \eqref{weyl}. We have
\begin{align*}
M_{\gamma,g}=e^{\frac{_{ \gamma Q}}{^2}\sigma(z)}\lim_{\epsilon\to 0}\epsilon^{\frac{_{\gamma^2}}{^2}}  e^{\gamma X_{{{g}},\epsilon}(z)}
d^{2}z.
\end{align*}
From this formula and the fact that $X_{{{e^{\varphi} g}}}\overset{law}{=}X_{{{g}}}-m_{e^{\varphi}g }(X_{{{g}}})$ we infer the Weyl transformation law
\begin{align*}
M_{\gamma,e^{\varphi} g}\overset{law}{=}e^{ \frac{_{ \gamma Q}}{^2}\varphi-\gamma m_{e^{\varphi}g }(X_g)}M_{\gamma, g}.
\end{align*}
Note that our definition of $\rho_{\gamma,g}$ so far depends on the choice of conformal coordinates. Let $\psi$ be a diffeomorphism. We want to define $\rho$ in different coordinates by
 \begin{align}\label{rhopsi}
 \rho_{\gamma, \psi^\ast g}:=\rho_{\gamma,  g}\circ\psi .
\end{align}
We have to check that this is well-defined, meaning that the above formula is consistent with \eqref{rhogamma} in the case when $\psi^* g$ is also a conformal metric. Hence suppose we have a metric $g$ and two different conformal coordinates are given by the diffeomorphisms $\psi_1$ and $\psi_2$. Thus we have $g = \psi_1^* g_1$ and $g=\psi_2^* g_2$ with
$$g_i = \tfrac{1}{2} e^{\sigma_i} (  dz \otimes d \bar z + d \bar z \otimes dz ).$$
From $\psi_{1}^{\ast}g_{1}=\psi_{2}^{\ast}g_{2}$ we get $g_1=(\psi_1^{-1})^* \psi_2^* g_1 = (\psi_2 \circ \psi_1^{-1})^* g_2$ which implies $\psi:= \psi_2 \circ \psi_1^{-1}$ is a conformal map and 
\begin{align*}
\sigma_1=\sigma_2\circ\psi+2\ln|\psi'|\,,
\end{align*}
where $\psi'$ denotes the complex derivative. Using the above formula and \eqref{h_conformal}  we conclude
\begin{align*}
(\rho_{\gamma,  g_{1}}\circ\psi_{1})(z)&= e^{\frac{\gamma^2}{4} (\sigma_1 \circ \psi_1)(z) + \frac{\gamma^2}{2} h_{g_1}(\psi_1(z),\psi_1(z))} \\
&= e^{\frac{\gamma^2}{4} ((\sigma_2\circ\psi_2)(z)+2\ln|\psi'(\psi_1(z))|) + \frac{\gamma^2}{2}(h_g(z,z) + \ln |\psi_1'(z)|)} \\
&=e^{\frac{_{ \gamma^2}}{^4}((\sigma_2\circ\psi_2)(z)+2(\ln|\psi_2'(z)| - \ln |\psi_1'(z)|)+\frac{_{ \gamma^2}}{^2}(h_{g_2}(\psi_2(z),\psi_2(z)) - \ln |\psi_2'(z)| + \ln |\psi_1'(z)|)} \\
&= e^{\frac{\gamma^2}{4}(\sigma_2 \circ \psi_2)(z) + \frac{\gamma^2}{2} h_{g_2}(\psi_2(z),\psi_2(z))} \\
&=(\rho_{\gamma,  g_{2}}\circ\psi_{2})(z)
\end{align*}
and this implies that \eqref{rhogamma} does not depend on the choice of conformal coordinates and \eqref{rhopsi} is consistent with it.

Now from Proposition \ref{gffinva} we infer
\begin{align}
\psi^{\ast}M_{\gamma,g}=
(\rho_{\gamma, g}\circ\psi)  \, m_{\gamma,\psi^{\ast}g}=M_{\gamma,\psi^\ast g}.
\label{Mdefi}
\end{align}

We can summarize these considerations to

\begin{proposition}\label{diffeandweyl} Let $\psi\in {\rm Diff}(\Sigma)$ and $\varphi\in C^{\infty}(\Sigma)$. Then we have the following equality of joint probability distributions 
\begin{align*}
(X_{g}\circ\psi, \psi^{\ast}M_{\gamma,g})\overset{law}{=}&
(X_{ \psi^{\ast}g}, M_{\gamma, \psi^{\ast}g})\\(X_{ e^{\varphi}g}, M_{\gamma, e^{\varphi}g})\overset{law}{=}&(X_{g}-c_{g}(\varphi)
, e^{ \frac{_{ \gamma Q}}{^2}\varphi-\gamma c_{g}(\varphi)
}M_{\gamma,g}),
\end{align*}
where the random variable $c_{g}(\varphi)$ is given by
\begin{align*}
c_{g}(\varphi)=
\tfrac{1}{ v_{e^{\varphi}g}(\Sigma)}\int X_{g}d v_{e^{\varphi}g }.
\end{align*}
\end{proposition}

\subsection{Liouville expectation}\label{liouville_expectation_section}

The GFF $X_g$ is almost surely an  element of  $H^{-s}(\Sigma, dv_g)$ for any $s>0$.
Let $F:H^{-s}(\Sigma,dv_g)\to\C$. 
The Liouville expectation of $F$, initially defined in Section 3 of \cite{DKRV}, is given by 
\begin{align}\label{lexpe}
\langle F \rangle_{\Sigma,g} &:=Z( \Sigma,g)\int_\R  \E \left[ F(c + X_g) e^{- \tfrac{Q}{4\pi} \int R_g(z) (c+X_g(z)) \, dv_g(z)  - \mu e^{\gamma c} M_{g,\gamma}(1) } \right] \, dc
\end{align}
where we use the notation $ M_{g,\gamma}(f)=\int_\Sigma fd M_{g,\gamma}$ so that $M_{g,\gamma}(1)$ denotes the total mass of the measure $M_{\gamma,g}$. The factor $Z( \Sigma,g)$ is the  ``partition function of the GFF'', explicitly 
\begin{align*}
Z( \Sigma,g)=e^{\hf \zeta_{\Sigma,g}'(0)}
v_g(\Sigma)^\hf 
\end{align*}
 where the zeta function is defined as
\begin{align*}
 \zeta_{\Sigma,g}(s)=\sum_{n=1}^\infty\lambda_{g,n}^s
\end{align*}
for real part of $s$ small enough and $\zeta_{\Sigma,g}'(0)$  is defined by analytic continuation, see Section 1 of \cite{Sarnak} for details. We include $Z(\Sigma,g)$ in the definition \eqref{lexpe} to match physics literature conventions. Especially this has the effect of shifting the central charge of the theory from $6Q^2$ to $1+6Q^2$, see Proposition \ref{weyl_anomaly}. For a diffeomorphism $\psi \in \rm{Diff}(\Sigma)$ the property $\lambda_{\psi^* g,n} = \lambda_{g,n}$ implies
\begin{align*}
Z( \Sigma,\psi^\ast g)=Z( \Sigma,g)
\end{align*}
and furthermore, equation (1.13) in \cite{Sarnak} gives
\begin{align*}
Z( \Sigma,e^\varphi g) &= e^{A(\varphi,g)} Z( \Sigma,g)
\end{align*}
where
\begin{align}\label{anomalyint}
A(\varphi,g) &= \tfrac{ 1}{96\pi} \int_{\Sigma} (|\nabla_g \varphi(z)|^2 + 2 R_g(z) \varphi(z) )\, dv_g(z).
\end{align}

\begin{proposition}\label{weyl_anomaly}
Suppose $F: H^{-s}(\Sigma,dv_g) \to \R$ is such that $\langle |F| \rangle_{\Sigma,g}<\infty$. Then we have the diffeomorphism covariance
\begin{align*}
\langle F \rangle_{\Sigma,\psi^* g} &=
\langle \psi_* F\rangle_{\Sigma, g}\,,
\end{align*}
where $(\psi_*F)(X) := F(X\circ\psi)$ 
and  the Weyl covariance
\begin{align}\label{weylcov}
\langle F \rangle_{\Sigma,e^\varphi g} &= e^{cA(\varphi,g)} \langle F( \cdot - \tfrac{Q}{2} \varphi) \rangle_{\Sigma,g}\,,
\end{align}
where $\varphi \in C^\infty(\Sigma,\R)$ and $c=1+6Q^2$.
\end{proposition}

\begin{proof}
The first claim follows directly from the identities
\begin{align*}
(X_{\psi^* g}, M_{\psi^*g,\gamma}(1) )& \overset{law}{=} (X_g \circ \psi,M_{g,\gamma}(1)) \,, \\
R_{\psi^* g} &= R_g \circ \psi
\end{align*}
which follow from Proposition \ref{diffeandweyl}, and the fact that $R_g$ is a scalar function (a $0$-form).
The second claim follows along the same lines as in \cite{DKRV}. For completeness we give the main steps.
Let $g'=e^{\varphi}g$. By  Proposition \ref{diffeandweyl}  and a shift $c'=c-c_g(\varphi)$ in the $c$ integral we have
\begin{align*}
Y_{g'}:=&\int_\R  \E \left[ F(c + X_{g'}) e^{- \tfrac{Q}{4\pi} \int_\Sigma R_{g'} (c+X_{g'}) \, dv_{g'}  - \mu e^{\gamma c} M_{g',\gamma}(1) } \right] \, dc\\=&\int_{\R} \E \big(F(c'+X_{g}) e^{- \frac{Q}{4\pi}\int_{\Sigma}R_{g'} (c'+X_{g})dv_{g'} - \mu e^{\gamma c' }M_{ g,\gamma}(e^{\frac{\gamma Q}{2}\varphi})}\big)dc'.
\end{align*}
Using $ R_{g'} v_{g'}=(R_{g} -\Delta_{g} \varphi)v_{g}$ and dropping the prime from $c$ this becomes
\begin{align*}
Y_{g'}=\int_{\R} \E \big(e^{ \frac{Q}{4\pi}
(X_g,\Delta_g\varphi)_g} F(c+X_{g}) e^{- \frac{Q}{4\pi}\int_{\Sigma}R_{g}(c+ X_{g})dv_{g} - \mu e^{\gamma c }M_{ g,\gamma}(e^{\frac{\gamma Q}{2}\varphi})}\big)dc.
\end{align*}
Next we apply the Girsanov theorem  to the factor $e^{ \frac{Q}{4\pi}
(X_g,\Delta_g\varphi)_g}$. Denoting the rest of the integrand by $H(X_{g},M_{g,\gamma})$ we have
\begin{align*}
\E (H(X_{g},M_{g,\gamma})
e^{ \frac{Q}{4\pi}
(X_g,\Delta_g\varphi)_g})=e^{\frac{Q^2}{32\pi^2}
(\Delta_g\varphi,G_{g}\Delta_g\varphi)_{g}}\E (H(X_{g}+h, e^{\gamma h}M_{g,\gamma})),
\end{align*}
with
\begin{align*}
h= \tfrac{Q}{4\pi}G_{g}\Delta_g\varphi.
\end{align*}
From \eqref{deltag} we obtain
$G_g\Delta_g\varphi=-2\pi\varphi  +\tilde c
$ with  $\tilde c=\tfrac{Q}{2}\frac{\int\varphi dv_g}{\int dv_g}$ so that 
$h=-\tfrac{Q}{2}\varphi +\tilde c$ and thus
\begin{align*}
Y_{g'}=e^{\frac{Q^2}{32\pi^2}
(\Delta_g\varphi,G_{g}\Delta_g\varphi)_{g}}\int_{\R} \E \big(F(c+\tilde c+X_{g}-\tfrac{Q}{2}\varphi) e^{- \frac{Q}{4\pi}\int_{\Sigma}R_{g} (c+\tilde c +X_{g}-\tfrac{Q}{2}\varphi)dv_{g} - \mu e^{\gamma (c+\tilde c) }M_{ g,\gamma}(1)}\big).
\end{align*}
After a shift in the $c$-integral we obtain 
\begin{align*}
Y_{g'}=e^{B(\varphi,g)}\int_{\R} \E \big(F(c+X_{g}-\tfrac{Q}{2}\varphi) e^{- \frac{Q}{4\pi}\int_{\Sigma}R_{g} (c +X_{g})dv_{g} - \mu e^{\gamma c}M_{ g,\gamma}(1)}\big),
\end{align*}
with
$$
B(\varphi,g)= \tfrac{Q^2}{32\pi^2}(\Delta_g\varphi,G_{g}\Delta_g\varphi)_{g}+ \tfrac{Q^2}{8\pi}\int_\Sigma R_{g}(z) \varphi(z) dv_g(z).
$$
The claim follows since 
$$(\Delta_g\varphi,G_{g}\Delta_g\varphi)_{g}=(\varphi,\Delta_gG_{g}\Delta_g\varphi)_{g}=-2\pi(\varphi,\Delta_g \varphi)_g. 
$$
\end{proof}

\subsection{Liouville correlation functions}

Choose a local conformal coordinate. We define
\begin{align}\label{rhoalpha}
 \rho_{\alpha,g}(z):=e^{\frac{_{ \alpha^2}}{^4}\sigma(z)+\frac{\alpha^2}{2}h_{g}(z,z)}%
\end{align}
and we define the regularized vertex operators by
\begin{align*}
V_{\alpha,g,N}(z) &= e^{\alpha(c+ X_{g,N}(z)) - \frac{\alpha^2}{2} \E X_{g,N}(z)^2} \rho_{\alpha, g}(z).
\end{align*}
Again we set
 \begin{align*}
 \rho_{\gamma, \psi^\ast g}=\rho_{\gamma,  g}\circ\psi,
\end{align*}
which is well defined by the same argument as with $\alpha=\gamma$ earlier. Hence
\begin{align}\label{Vdiffeo}
V_{\alpha,g,N}\circ\psi &= V_{\alpha,\psi^\ast g,N}.
\end{align}

\begin{proposition}\label{correlationf} The correlation functions
\begin{align*}
\langle \prod_{i=1}^n V_{\alpha_i,g}(z_i) \rangle_{\Sigma,g} &:= \lim_{N \to \infty} \langle \prod_{i=1}^n V_{\alpha_i,g,N}(z_i) \rangle_{\Sigma,g}\,,
\end{align*}
exist and are non-zero if and only if the $\alpha_i$'s satisfy the Seiberg bounds
\begin{align}\label{seibe}
\sum_{i=1}^n \alpha_i > Q\chi(\Sigma)\,, \quad \alpha_i < Q \quad \forall i \,,
\end{align}
where $\chi(\Sigma) = 2 - 2 \rm{genus}(\Sigma)$ is the Euler characteristic. Furthermore they satisfy the diffeomorphism and Weyl transformation laws \eqref{diffeo} and \eqref{weyl} with $\Delta_\alpha = \frac{\alpha}{2}(Q-\frac{\alpha}{2}).$

\end{proposition}

\begin{proof} The strategy for the proof of convergence is the following. We first switch the integration order and argue that the $c$-integral converges. Then we evaluate the $c$-integral which then yields a different representation of the correlation function in terms of an expectation of a moment of a GMC integral with no $c$-integral remaining. Consider \eqref{lexpe} with $F(X)=e^{(X,f)_g}$ with $f \in C^\infty(\Sigma)$. For the scalar curvature the Gauss--Bonnet theorem takes the form $\int R_gdv_g= 4 \pi \chi(\Sigma)$,\footnote{Note that we are using the scalar curvature $R_g$, which is twice the \emph{Gaussian curvature} $K_g$.} so we get
\begin{align}\nonumber
\langle F\rangle_{\Sigma,g} &=Z( \Sigma,g)\int_\R e^{c((f,1)_g-Q\chi(\Sigma))} \E \left[  e^{((f- \tfrac{Q}{4\pi} R_g),X_g)_g}e^{ - \mu e^{\gamma c} M_{g,\gamma}(1) } \right] \, dc\\&=Z( \Sigma,g) \E \left[  e^{(h,X_g)_g}\int_\R e^{c((f,1)_g-Q\chi(\Sigma))}e^{ - \mu e^{\gamma c} M_{g,\gamma}(1) }\, dc \right] \label{lexpe1}
\end{align}
where we used Fubini's theorem and defined
$$
h=f- \tfrac{Q}{4\pi} R_g.
$$
Since $M_{g,\gamma}(1)>0$ almost surely, the $c$-integral converges provided
\begin{align}\label{lexpe2}
(f,1)_g>Q\chi(\Sigma),
\end{align}
and after evaluating the $c$-integral we get 
\begin{align*}
\langle F\rangle_{\Sigma,g}&=
\gamma^{-1}{\mu^{-s_f}}\Gamma(s_f) Z( \Sigma,g) \E\,\big[ e^{(h,X_g)_g  } M_{ \gamma,g}(1)^{-s_f}\big],
\end{align*} 
where
\begin{align*}
s_f=\tfrac{(f,1)_g-Q\chi(\Sigma)}{\gamma}.
\end{align*} 
Finally a shift in the  Gaussian integral (Girsanov theorem) gives
\begin{align*}
\langle F\rangle_{\Sigma,g}&=
\gamma^{-1}{\mu^{-s_f}}\Gamma(s_f) Z( \Sigma,g) e^{\hf (h,G_g h)_g}\E\,\big[  M_{ \gamma,g}(e^{\gamma G_gh})^{-s_f}\big].
\end{align*} 
For the correlation functions we take
\begin{align*}
F(X) &= e^{(X_g,f)_g - \sum_{i=1}^n ( \frac{\alpha_i^2}{2} \E X_{g,N}(z_i)^2 + \ln \rho_{g,\alpha_i}(z_i))} 
\end{align*}
with $f=\sum_{i=1}^n  \alpha_i\sum_{n=0}^Ne_{g,n}e_{g,n}(z_i)$, because then $(X_g,f)_g = \sum_{i=1}^n  \alpha_i X_{g,N}(z_i)$. Then, the condition \eqref{lexpe2} becomes the first of the conditions \eqref{seibe}.  As $N\to\infty$ in a neighborhood of $z_i$
$$
e^{\gamma G_gh(z)}=|z-z_i|^{-\gamma\alpha_i}+\caO(1)
$$
and the condition $\alpha_i<Q$  is needed for the $M_{ \gamma,g}$ integrability of this singularity, see Lemma 3.3 in \cite{DKRV}.

Diffeomorphism covariance follows from \eqref{Vdiffeo} in the limit. 

For the Weyl covariance let us again choose conformal coordinates around the insertion points $z_i$ and use the circle average regularization
\begin{align*}
\langle \prod_{i=1}^n V_{\alpha_i,g}(z_i) \rangle_{\Sigma,g} &:= \lim_{\epsilon \to 0} \langle \prod_{i=1}^n V^\epsilon_{\alpha_i,g}(z_i) \rangle_{\Sigma,g}\,,
\end{align*}
where
\begin{align}\label{Veps}
V^\epsilon_{\alpha,g}(z)=e^{ \frac{_{ \alpha^2}}{^4}\sigma}\epsilon^{\frac{\alpha^2}{2}}e^{\alpha(c+ X_{g,\epsilon}(z))}.
\end{align}
By \eqref{weylcov} we have then
\begin{align*}
 \langle \prod_{i=1}^n V^\epsilon_{\alpha_i,e^\varphi  g}(z_i) \rangle_{\Sigma,e^\varphi  g}=e^{cA(\varphi,g)}\prod_{i=1}^n e^{ \frac{_{ \alpha_i^2}}{^4}\varphi(z_i)-\hf Q\alpha_i\varphi_\epsilon(z_i)} \langle \prod_{i=1}^n V^\epsilon_{\alpha_i,g}(z_i) \rangle_{\Sigma,g}
\end{align*}
which implies the claim since $\varphi_\epsilon(z_i)\to \varphi(z_i)$ as $\epsilon\to 0$. \end{proof}

\section{Conformal Ward Identities}\label{ward_section}

We will now specialize to the case $\Sigma=\mathbb{S}^2=\hat\C$ and consider the metric dependence of the vertex correlation functions:
\begin{align*}
g\mapsto F(g,\bx)= \langle\prod_{i=1}^{N}V_{\alpha_{i},g}(x_{i})\rangle_{g}
\end{align*}  
where from now on we drop the $\Sigma$ from the notation. Our objective is to construct the derivatives \eqref{Tcorre} and prove the identities \eqref{wardid}. The identities \eqref{diffeo} and \eqref{weyl} make the dependence on the metric quite explicit since the sphere has only one conformal class (see below for the definition), a fact we will recall next. 

\subsection{Beltrami equation}\label{beltrami_section}
Let $\caM$ be the set of smooth metrics in $\hat \C$.  We may work on the coordinate chart $\C$ and identify $g\in \caM$ with a smooth function $g(z)=\{g_{\alpha\beta}(z)\}$ taking values in positive matrices such that $D\zeta^T (g\circ\zeta) D\zeta$ is smooth as well where $\zeta(z)=z^{-1}$ (this means that $g$ is smooth at infinity).  Similarly 
$\psi\in {\rm Diff}(\hat \C)$ can be identified with a diffeomorphism $\psi\in {\rm Diff}( \C)$ satisfying the additional condition $\psi\circ\zeta\in {\rm Diff}(\C)$.

The sphere has only one conformal class of metrics, which means that for a fixed metric $\hat g \in \mathcal{M}$, any other metric $g \in \mathcal{M}$ can be written as $g=e^\varphi \psi^* \hat g$, where $\varphi \in C^\infty(\rs)$ and $\psi \in \opn{Diff}(\rs)$. We aim to prove the conformal Ward identities by varying the metric $\hat g$ and seeing how this affects the correlation function $\langle \prod_{i=1}^N V_{\alpha_i,\hat g}(z_i) \rangle_{\hat g}$. To this end, we want to compute how $\varphi$ and $\psi$ depend on the perturbed metric $g$, or more concretely, if $g$ is roughly of the form $g=\hat g + \epsilon f$, what is the $f$ dependency of $\varphi$ and $\psi$ in the first order in $\epsilon$. The purpose of this section is to find the relation between the perturbation $f$ and the functions $\psi$ and $\varphi$, and we will see that the equation $g=e^{\varphi} \psi^* \hat g$ will lead us to the Beltrami equation.

Let $g\in\caM$ and set
\begin{align}
\hat g_{\alpha\beta}(z)=
e^{\sigma(z)} \delta_{\alpha\beta} ,
\label{hatgee}
\end{align}  
where $\sigma \in C^\infty(\rs)$ and $\delta$ denotes the Euclidean metric. Let us look for a function  $\varphi$ and  a diffeomorphism $\psi$ such that
\begin{align}\label{pertmet1}
g &= e^{\varphi} \psi^* \hat g=  
e^{\varphi + \sigma \circ \psi} (D \psi)^T D \psi\,,
\end{align}
where the second equality comes from \eqref{pullback}. By taking determinants we get
\begin{align}\label{detg}
\det g = (e^{\varphi + \sigma \circ \psi} \det D \psi)^2\,,
\end{align}
and by plugging this back into the equation we get
\begin{align*}
(D \psi)^T D \psi =  \gamma \det D \psi
\end{align*}
where
\begin{align*}
\gamma := \tfrac{g}{\sqrt{\det g}}.
\end{align*}
Hence $D \psi^T = \gamma (D \psi)^{-1} \det D \psi$ which in complex coordinates becomes the Beltrami equation (see Theorem 10.1.1. in \cite{astala})
\begin{align}\label{beltreq}
\partial_{\bar z} \psi = \mu \partial_z \psi\,,
\end{align}
where 
\begin{align}\label{belcoef}
 \mu := \tfrac{\gamma_{\bar z \bar z} }{\hf + \gamma_{z \bar z}}\,.
\end{align}
It is readily checked that
\begin{align*}
\| \mu \|_\infty < 1.
\end{align*}
Indeed, we have
\begin{align*}
|\mu|^2 &= \frac{\gamma_{\bar z \bar z} \overline{\gamma_{\bar z \bar z}}}{(\tfrac{1}{2} + \gamma_{z \bar z}) ( \tfrac{1}{2} + \overline{\gamma_{z \bar z}})} = \frac{\gamma_{z \bar z}^2 - \tfrac{1}{4} }{\gamma_{z \bar z}^2 +\gamma_{z \bar z} + \tfrac{1}{4} }\,.
\end{align*}
Above we used  \eqref{det} (note that $\overline{\gamma_{\bar z \bar z}} = \gamma_{zz}$ always in the complex coordinates), $\det \gamma =1$, and the fact that $\gamma_{z \bar z}$ is always real (to simplify the denominator) and positive, which follows from the facts that $\gamma$ is positive definite and $4\gamma_{z \bar z} =  \opn{tr} \gamma$. Now $\|\mu\|_\infty < 1$ follows.

The Beltrami equation is solved by writing 
\begin{align}\label{psiu}
\psi(z) = z +  u(z)
\end{align}
whereby \eqref{beltreq} becomes
\begin{align}\label{beltreq1}
\partial_{\bar z} u - \mu \partial_z u &= \mu  \,.
\end{align}
To solve this recall the Cauchy transform $\mathcal{C}:C_0^\infty(\C)\to C^\infty(\C)$
\begin{align*}
(\mathcal{C} f)(z) &:= \tfrac{1}{\pi} \int_\C \frac{f(z')}{z-z'} \, d^2z'\,,
\end{align*}
and the Beltrami transform  $\mathcal{B}:C_0^\infty(\C)\to C^\infty(\C)$  given by $\mathcal{B} := \partial_z \mathcal{C} = \mathcal{C} \partial_z$ (see Chapter 4 of \cite{astala} for basic properties of these integral transforms). We have $\partial_{\bar z}\caC f = \mathcal{C} \partial_{\bar z} f=f$  so that \eqref{beltreq1} can be written for $u\in C^\infty(\C)$ as
\begin{align}
(1-\mathcal{C} (\mu \partial_z) )u &= \mathcal{C} \mu
\end{align}
and then as a Neumann series
\begin{align}\label{neumann}
 u &= (1-\mathcal{C} \mu \partial_z)^{-1} \mathcal{C} \mu  = \sum_{n=0}^\infty (\mathcal{C} \mu \partial_z)^n \mathcal{C} \mu  = \sum_{n=0}^\infty  \mathcal{C} v_n 
\end{align}
where
$$
v_n:=(\mu \mathcal{B})^n \mu.
$$
We will also denote
\begin{align}\label{u_n}
u_n &:= \mathcal{C}v_n\,.
\end{align}
The convergence of this series is classical and for what follows we state it in a slightly more general setup for a smooth family of Beltrami coefficients $\mu(z,\epsilon)$. Our proof is a slight variation of the proof \cite{astala}, Sections 5.1 and 5.2, so we will be brief. For a multi-index $l=(l_1,l_2,l_3)$ let $D^l=\partial_\eps^{l_1}\partial_z^{l_2}\partial_{\bar z}^{l_3}$.

\begin{proposition}\label{beltrami_solu}
Let $\mu\in C_0^\infty(\C\times\R)$ with $k:=\|\mu\|_{L^\infty(\C \times \R)} < 1$. Then the series \eqref{neumann} converges uniformly together with all its derivatives and setting $u_n=\caC v_n$ we have
\begin{subequations}
\label{Du}
\begin{align}
 |D^lu(z,\eps)| & \leq C_l (1+|z|)^{-1}\,,  \\
 |D^lu_n(z,\eps)| & \leq C_l (1+|z|)^{-1}. 
\end{align}
\end{subequations}
Furthermore, for each $\epsilon$ the function $\psi(z, \epsilon)=z+u(z, \epsilon)$ defines a smooth diffeomorphism of $\hat\C$.
\end{proposition}
\begin{proof}
First, Theorem 4.5.3 in \cite{astala} implies $S_p:=\|\caB\|_{L^p(\C)\to L^p(\C)}\to1$ as $p\to 2$ with $S_2=1$. From this we obtain
$$\|\mu\caB\|_{L^p(\C)\to L^p(\C)}\leq k S_p<1$$    for $p$ close enough to $2$. For the rest of the proof we fix any such $p$, say $p=2$.\footnote{Note the slight difference between our proof and the proofs in Sections 5.1 and 5.1 of \cite{astala}. In \cite{astala} $p>2$ is assumed to establish integrability properties of $u$ (see e.g. Lemma 5.2.1). For us the estimate \eqref{Du} is enough (and implies $u(\cdot,\eps) \in W^{k,p}(\C)$ for $p>2$) so we can fix $p=2$.} It follows that the series $\sum_{n=0}^\infty (\mu\mathcal{B})^n\mu$ converges in $ L^2(\C)$ uniformly in $\epsilon$. 
Next we show that the limit belongs to the Sobolev space $W^{k,2}(\C)$ for all $k$ and $\epsilon$. By applying the product rule of the derivative and $\partial_z \mathcal{B} \mu = \mathcal{B} \partial_z \mu$, $\partial_{\bar z} \mathcal{B} \mu = \mathcal{B} \partial_{\bar z} \mu$, we get for $v_n=(\mu \mathcal{B})^n \mu$
\begin{align*}
\|D^lv_n\|_2=\Big\|\sum_{\sum_{i=1}^{n+1} k_i=l} D^{k_1}\mu\caB D^{k_2}\mu\caB\dots D^{k_{n+1}}\mu\Big\|_2\,.
\end{align*}
By using $\|D^{k_i} \mu  \mathcal{B} D^{k_j} \mu \|_2 \leq \|D^{k_i} \mu\|_\infty S_2 \| D^{k_j} \mu\|_2$, we get the upper bound
\begin{align*}
\Big\|\sum_{\sum_{i=1}^{n+1} k_i=l} D^{k_1}\mu\caB D^{k_2}\mu\caB\dots D^{k_{n+1}}\mu\Big\|_2\leq C(l,n)(kS_2)^{n-\|l\|}
\end{align*}
where the constant $C(l,n)$ depends on $\|D^k\mu\|_\infty$ for $|k|\leq |l|$. This shows $z \mapsto v(z,\epsilon):=\sum_{n=0}^\infty v_n(z,\epsilon)$ is in the Sobolev space $W^{k,2}(\C)$ for all $k$ and $\epsilon$, and since $\mu(\cdot,\eps) \in C_0^\infty(\C)$ we conclude $v_n(\cdot,\eps) \in C^\infty_0(\C )$ and  $v(\cdot,\eps) \in C_0^\infty(\C )$. For the bounds \eqref{Du} recall that $u_n = \mathcal{C} v_n$ and $u = \mathcal{C} v$. Since $\mathcal{C}$ maps $C^\infty_0(\C)$ into $C^\infty(\C)$, the functions $z \mapsto u_n(z,\epsilon)$ and $z \mapsto u(z,\epsilon)$ belong to $C^\infty(\C)$ and are bounded on compact sets. Now we get the bound
\begin{align*}
| D^l u(z,\epsilon)| & \leq \frac{1}{\pi} \int_{\opn{supp}(v(\cdot,\epsilon))} \frac{|D^l v(y,\eps)|}{|z-y|} \, d^2y \leq C \frac{1}{1+|z|}\,,
\end{align*}
where $\opn{supp} v$ denotes the support of $v$. Same argument yields the same bound for $u_n$.
Now we have shown that $\psi(z,\eps)=z+u(z,\eps)$ belongs to $C^\infty(\C)$ for each $\eps$ and $\psi(z,\eps) = z + \mathcal{O}(1/z)$ as $z \to \infty$. Theorems 5.2.3 and 5.2.4 in \cite{astala} then imply that $\psi(z,\eps) = z+u(z,\eps)$ is a smooth diffeomorphism of $\rs$ for each $\epsilon$.
\end{proof}

Next we want to argue that the function $\varphi$ in \eqref{pertmet1} is smooth. Equation \eqref{detg} implies that $\varphi$ is given by
\begin{align*}
e^{\varphi}=e^{-\sigma\circ\psi}\frac{\sqrt{\det g}}{\det D\psi}=e^{-\sigma\circ\psi}
\frac{\sqrt{\det g}}{|1+\partial_{z}u|^{2}-|\partial_{\bar z}u|^{2}}
\end{align*}
where we used 
$\det D\psi=|\partial_{z}\psi|^{2}-|\partial_{\bar z}\psi|^{2}$ and $\psi(z)=z+u(z)$.
Since 
 $g$ is a metric on the Riemann sphere $\hat\C$, the volume density $\sqrt{\det g}$ must be smooth at infinity. This means that $\sqrt{\det \zeta^* g}$ is a smooth (and positive) function at the origin, where $\zeta(z) = \tfrac{1}{z}$. Thus we can write $\sqrt{\det \zeta^* g} = e^{\rho(z,\epsilon)}$ where $\rho \in C^\infty(\C \times \R)$. We have
 \begin{align*}
 \sqrt{\det \zeta^* g} &= \sqrt{\det ( D \zeta (g \circ \zeta) D \zeta^T)} = \det D \zeta \sqrt{\det g \circ \zeta} = \tfrac{1}{|z|^4} \sqrt{\det g \circ \zeta}\,.
\end{align*}
Thus we get
\begin{align*}
\ln \sqrt{\det g(z,\epsilon)}=-4\ln|z|+\rho(1/z,\epsilon).
\end{align*}
On compact sets this function is bounded, and as $z \to \infty$ the absolute value is dominated by $C|z|^{-1}$. Thus the bound \eqref{Du} holds when we replace $u$ by $\ln \sqrt{\det g}$ or $\sigma$ (by the same argument). We conclude then 
\begin{align}\label{Duvarp}
 |D^l\varphi(z)|\leq C_l  (1+|z|)^{-1}.
\end{align}
Thus $\varphi(\cdot,\epsilon) \in W^{k,p}(\C )$ for all $k$ and $p>2$ and we conclude $\varphi \in C^\infty(\C \times \R)$.

\subsection{The stress-energy tensor}\label{se_section}

In this section we give the precise definition of the derivatives \eqref{Tcorre}.
\begin{definition}\label{functional_derivative}
Let $S: C^\infty(\C) \to \C$ be a functional (not necessarily linear) such that for all $h \in C^\infty(\C)$ the function $\eps \mapsto S(h+\eps f)$ is differentiable at $0$ for all $f \in C^\infty_0(\C)$. If it also holds that the derivative is linear and continuous in $f$, we denote by $\frac{\delta S}{\delta h} \in \mathcal{D}'(\C)$ the distribution
\begin{align*}
(\tfrac{\delta S}{\delta h}, f) &:= \de S(h+\eps f)\,.
\end{align*}
These derivatives are also known as the Gateaux differentials. If $\tfrac{\delta S}{\delta h}$ can be realized as a continuous function, that is, if we have
\begin{align*}
(\tfrac{\delta S}{\delta h}, f) &= \int_\C s(z) f(z) \, d^2z
\end{align*}
for some $s \in C(\C)$, then we define $\frac{\delta S}{\delta h(z)} := s(z)$.
\end{definition}

The Liouville correlation functions \eqref{correlat} are functions of the metric $g$. As was explained in the beginning of Section \ref{ward_section}, we identify $g$ with a positive matrix valued smooth function $g(z)= \{g_{\mu \nu}(z)\}$ on $\C$ such that also $D \zeta^T (g \circ \zeta) D \zeta$ is smooth, where $\zeta(z) = z^{-1}$. We denote the inverse matrices by $\{g^{\mu \nu}(z)\}$. In this picture the correlation functions are functions of the smooth functions $g_{\mu \nu }$ and by perturbing these functions we may compute functional derivatives of the form 
\begin{align*}
\frac{\delta}{\delta g^{\mu \nu}} \langle \prod_{i=1}^N V_{\alpha_i,g}(x_i) \rangle_g\,,
\end{align*}
where this derivative is to be understood in the way that was explained above. Note that we have to show that the derivatives
\begin{align*}
\de \langle \prod_{i=1}^N V_{\alpha_i,g_\eps}(x_i) \rangle_{g_\eps}
\end{align*}
exist, where $g_\eps^{\mu \nu} = g^{\mu \nu} + \eps f^{\mu \nu}$, $f^{\mu \nu} \in C^\infty_0(\C)$. Once the existence of these derivatives is established, they are shown to be given by distributions evaluated at the perturbation functions $f̂^{\mu \nu}$. Then we prove that these distributions can be represented by functions, which will be denoted by $\langle T_{\mu \nu}(z) \prod_{i=1}^N V_{\alpha_i,g}(z_i) \rangle_g$. 

Let
$$ \C_\bx := \C \setminus\{x_1,\hdots,x_N\}$$
and let
$f\in C_0^\infty(\C_\bx, M_2)$ be a smooth function with compact support in $\C_\bx$ and taking values in symmetric  $2\times 2$ matrices. 
Consider the perturbed  metric $g(z,\epsilon)$ with the inverse given by
\begin{align}\label{metricvar}
g^{\mu \nu} &= \hat g^{\mu \nu}+\epsilon f^{\mu \nu}
\end{align}
where $\hat g$ is given by \eqref{hatgee} so that $ \hat g^{\mu \nu}=e^{-\sigma}\delta^{\mu \nu}$. It follows that
$g$ defines a metric (i.e. is positive definite) if $\epsilon$ is small enough. We will also use the notation
\begin{align}
F(\hat g, \eps) &:= \langle \prod_{i=1}^N V_{\alpha_i, g}(x_i) \rangle_g\,.
\end{align}
Then we have


\begin{proposition}\label{smoothness}
The function $F(\hat g,\cdot): \epsilon\mapsto  \langle\prod_{i=1}^{N}V_{\alpha_{i},g}(x_{i})\rangle_{g}$ is smooth in a neighbourhood of the origin.\footnote{The function $F$ also depends on the points $x_i$, but we omit this from the notation.} Furthermore, for any positive integer $n$ we have  
 \begin{align}
 \partial_\epsilon^{n}F(\hat g, \eps)|_{\epsilon=0}=\caT_{n}(f,\dots,f)
\label{actiooo}
\end{align}
 where the $n$-linear function $\caT_{n}$ defines a distribution $\caT_{n}\in (C_0^\infty(\C_\bx, M_2)^{\otimes n})^*$.\footnote{That is, $\caT_n$ is a complex valued continuous linear map, taking as arguments functions from the $n$-fold tensor product $\bigotimes_{k=1}^n C^\infty_0(\C_x,M_2)$. We use the notation $\mathcal{T}_n(f_1,\hdots,f_n) := \mathcal{T}_n (f_1\otimes\hdots\otimes f_n)$.} 

Furthermore, let $f_i\in C_0^\infty(\C_\bx, M_2)$, $i=1,\dots, n$ have disjoint compact supports 
and set
$$g^{\mu\nu}=\hat g^{\mu\nu}+\sum_{i=1}^n\epsilon_{i}f_{i}^{\mu\nu}.$$
Then for the function $F(\hat g, \cdot): (\epsilon_1,\hdots,\epsilon_n) \mapsto \langle\prod_{i=1}^{N}V_{\alpha_{i},g}(x_{i})\rangle_{g}$ we have
\begin{align}
\prod_{i=1}^n \partial_{\epsilon_i}  F(\hat g, \eps_1,\hdots,\eps_n) |_{(\epsilon_1,\hdots,\epsilon_n)=0} =  \caT_{n}(f_{1},\dots,f_{n})
&=\int  F^{\hat g}_{\mu_{1} \nu_1\dots \mu_n \nu_{n}}(z_{1},\dots,z_{n})\prod_{i=1}^{n}  f_i^{\mu_{i}\nu_{i}}(z_{i})dv_{\hat g}(z_{i})\,,
\label{tandf}
\end{align}
where we sum over repeated indices and $F^{\hat g}_{\mu_{1} \nu_{1}\dots \mu_{n}\nu_{n}}(z_{1},\dots,z_{n})$ are smooth functions  in the region $z_{i}\in\C_\bx$ with $z_{i}\neq z_{j}$ when $i\neq j$.
\end{proposition}

\begin{proof}
By Proposition \ref{beltrami_solu} for $\epsilon$   small enough we have
\begin{align*}
g &= e^{\varphi} \psi^* \hat g\,.
\end{align*}
Writing 
\begin{align}\label{geepsi}
g=e^\sigma(\delta+\zeta)
\end{align}
the
Beltrami coefficient \eqref{belcoef} is
\begin{align}\label{muepsi}
\mu 
=
\frac{\zeta_{\bar z\bar z}}{\sqrt{(1+4\zeta_{z\bar z})^{2}-4\zeta_{z z}\zeta_{\bar z\bar z}}+\zeta_{ z\bar z}}
\end{align}
and the function $\varphi$ is given by
\begin{align}\label{varphi}
\varphi &=  \sigma- \sigma \circ \psi +\hf\ln((1+4\zeta_{z\bar z})^{2}-4\zeta_{z z}\zeta_{\bar z\bar z})-\ln(|1+\partial_{z}u|^{2}-|\partial_{\bar z}u|^{2}).
\end{align}
Now by the diffeomorphism and Weyl transformation laws in Proposition \ref{correlationf} 
\begin{align}\label{fundaid}
F(\hat g, \eps)= \langle \prod_{i=1}^N V_{\alpha_i,g}(x_i) \rangle_{g}
&= \langle \prod_{i=1}^N V_{\alpha_i, e^{\varphi} \psi^* \hat g}(x_i) \rangle _{e^{\varphi} \psi^* \hat g} = e^{cA(\varphi, \psi^* \hat g)} e^{- \sum_{i=1}^N \Delta_{\alpha_i}\varphi(x_i)} \langle \prod_{i=1}^N V_{\alpha_i, \hat g}(\psi(x_i)) \rangle_{\hat g} \,.
\end{align}
We will now argue that smoothness of $F^{\hat g}$ in  $\epsilon$ will follow from smoothness of $\varphi$ and $u$ in $\epsilon$ (recall that $\psi(z)=z+u(z)$). First, to prove smoothness of the expectation on the right-hand side of \eqref{fundaid} we use the result  of the second author  \cite{Oik}  that the correlation function $(x_1,\hdots,x_N) \mapsto \langle \prod_{i=1}^N V_{\alpha_i, \hat g}(x_i) \rangle_{\hat g}$ is smooth in the region of non-coinciding points. Since $\psi$ is a diffeomorphism the points $\psi(x_i)$ are non-coinciding  as well and smoothness of the expectation in $\epsilon$ follows. 
Smoothness of the anomaly term follows from the bounds \eqref{Du} and \eqref{Duvarp}  which guarantee convergence of the integrals over $\C$.

We will now deduce from equations \eqref{psiu}, \eqref{varphi} and \eqref{fundaid} that to compute the derivative \eqref{actiooo} it is fundamental to compute $\partial_\epsilon^mu|_{\epsilon=0}$ and $\partial_\epsilon^m \zeta_{\mu \nu}|_{\epsilon=0}$ for $m\leq n$. Before starting the computations it is useful to remark that $\zeta$ can be written in terms of the functions $f^{\alpha \beta}$ by combining \eqref{metricvar} and \eqref{geepsi}. This yields
\begin{align}\label{zeta}
g&=\hat g + e^\sigma \zeta=(\hat{g}^{-1}+\epsilon f^{-1})^{-1} = \hat g  \sum_{k=0}^\infty (-\epsilon \hat g f^{-1})^k 
\end{align}
where $f^{-1}$ is the matrix with entries $\{f^{\alpha \beta}\}$ and $\hat g f^{-1}$ is a matrix product.
It follows that $\partial^m \zeta_{\mu \nu}|_{\epsilon=0}$ is a homogeneous polynomial of degree $m$ in the variables $\{f^{\alpha \beta}(z)\}$ and
\begin{align*}
\mu &=   \zeta_{\bar z \bar z} + \mathcal{O}(\epsilon^2) = - \tfrac{\epsilon}{4} e^{\sigma}  f^{zz} + \mathcal{O}(\epsilon^2)\,.
\end{align*}
Hence  $\mu=\mathcal{O}(\epsilon)$ and from \eqref{u_n} we get $u_{k}=\mathcal{C}(\mu \mathcal{B})^k \mu =\caO(\epsilon^{k+1})$. Now $u=\sum_k u_k$ implies
\begin{align*}
\partial_{\epsilon}^mu|_{\epsilon=0}=
\sum_{k=0}^{m-1}\partial_\epsilon^mu_{k}|_{\epsilon=0}\,.
\end{align*}
From \eqref{muepsi} and \eqref{zeta} we have
\begin{align*}
\partial_\epsilon^l\mu(z)|_{\epsilon=0}=p_l(z)
\end{align*}
where $p_l(z)$ is a homogeneous polynomial of degree $l$ in the variables $\{f^{\alpha\beta}(z)\}$. In particular $p_l\in C_0^\infty(\C_\bx)$. Hence 
\begin{align}\label{u_k_derivative}
\partial_\epsilon^m u_{k}|_{\epsilon=0}=\sum_{\sum_{i=1}^{k+1} l_i=m}\caC p_{l_1}\caB p_{l_2}\dots\caB p_{l_{k+1}}.
\end{align}
Now we know how the basic terms $\partial_\epsilon^m u|_{\epsilon=0}$ and $\partial_\epsilon^m \zeta_{\mu \nu}|_{\epsilon=0}$ look like.

Let us  now look at the various contributions to the derivative \eqref{actiooo}. From \eqref{fundaid} we see that we get derivatives of the form
\begin{align}\label{derivatives}
&\partial_\epsilon^k A(\varphi,\psi^* \hat g)|_{\epsilon=0} \,, \quad  \partial_\epsilon^k \varphi(x_i) |_{\epsilon=0}\,, \quad \partial_\epsilon^k \langle \prod_{i=1}^N V_{\alpha_i,\hat g}( \psi(x_i)) \rangle_{\hat g} |_{\epsilon=0},
\end{align}
with $k \leq n$. The anomaly term equals
\begin{align*}
A(\varphi,\psi^* \hat g) &= \tfrac{1}{96\pi} \int_\C (\psi^* \hat g)^{\alpha \beta} \partial_\alpha \varphi \partial_\beta \varphi dv_{\psi^* g} + \tfrac{1}{48 \pi}\int_\C R_{\psi^* \hat g} \varphi dv_{\psi^* g}\,.
\end{align*}
By recalling that $(\psi^* \hat g)^{\alpha \beta} = (D \psi (\hat g \circ \psi) D\psi^T)^{\alpha \beta}$, $dv_{\psi^* \hat g} = |\det D \psi| dv_{\hat g} $ and $R_{\psi^* \hat g} = R_{\hat g} \circ \psi$ we see that to compute $\partial_\epsilon A(\varphi, \psi^* \hat g)$ it suffices to compute $\partial_\epsilon \varphi$ and $\partial_\epsilon \psi = \partial_\epsilon u$. For the derivatives $\partial_\epsilon^k \varphi(x_i)$ it suffices to know $\partial_\epsilon^k u$ because of $\eqref{varphi}$ and $\eqref{zeta}$. Finally, for the last term in \eqref{derivatives} we note that
\begin{align*}
\partial_{\epsilon} \langle \prod_{i=1}^N V_{\alpha_i,\hat g}(\psi(x_i)) \rangle_{\hat g}&= \sum_{j=1}^N \partial_\epsilon \psi (x_j) \partial_{x_j} \langle \prod_{i=1}^N V_{\alpha_i,\hat g}(x_i) \rangle_{\hat g} = \sum_{j=1}^N \partial_\epsilon u (x_j) \partial_{x_j} \langle \prod_{i=1}^N V_{\alpha_i,\hat g}(x_i) \rangle_{\hat g}\,.
\end{align*}
Thus, we see that all the derivatives in \eqref{derivatives} reduce to computing derivatives of $u$ and derivatives of $\zeta_{\mu\nu}$. The derivatives of $u$ we already computed above and the derivatives of $\zeta$ are easily read off of \eqref{zeta}.

Next we argue that the left-hand side of \eqref{actiooo} can really be expressed in terms of a distribution. Let $l$ be a positive integer and fix a set of positive integers $(l_i)_{i=1}^k$ such that $\sum_i l_i = l$. Let $m^i_{l_i}(f)$ denote a monomial of degree $l_i$ in  the variables $\{D^Kf^{\alpha\beta}(x)\}$ where $D=\partial_x,\partial_{\bar x}$ and  $K\leq 2$. We will now explain that the above and \eqref{varphi} imply that the derivative \eqref{actiooo} consists of products and sums of terms that are $l$-linear functionals of $\{f^{\alpha\beta}(x)\}$ of the form
\begin{subequations}\label{Tpl}
\begin{align}
\int_\C \caT_{p,\bl}(z,f,\dots,f) \, dv_g(z)&= \int_\C (\caC^p m^1_{l_1}(f)\caB m^2_{l_2}(f)\dots\caB m^k_{l_k}(f))(z) \, dv_g(z)\,, \\
\caT_{p,\bl}(x_i,f,\dots,f)&=(\caC^p m^1_{l_1}(f)\caB m^2_{l_2}(f)\dots\caB m^k_{l_k}(f))(x_i)\,,
\end{align}
\end{subequations}
with $p=0,1$ and $\bl=(l_1,\hdots,l_k)$. The fact that the above expression can be written in terms of an $l$-linear functional $\caT_{p,\bl}$ follows from the definition of $m^i_{l_i}$. The derivatives containing the integral over $z$ arise from taking derivatives of $A(\varphi,\psi^* \hat g)$ in \eqref{fundaid} and the terms with the $x_i$'s come from the derivatives of $\varphi(x_i)$ and $\langle \prod_{i=1}^N V_{\alpha_i, \hat g} (x_i) \rangle_{\hat g}$. Then the expressions \eqref{Tpl} come from the observation above that everything reduces to derivatives of $u$ and $\zeta$ and using \eqref{zeta} and \eqref{u_k_derivative}. The maps $f\mapsto m^i_{l_i}(f)$ and $(f,g)\mapsto f\caB g$ are continuous maps $C_0^\infty\to C_0^\infty$ and  $C_0^\infty\times C_0^\infty\to C_0^\infty$, respectively (in the Fr\'{e}chet topology of $C_0^\infty$), and  $f\mapsto (\caC f)(z)$ is continuous $C_0^\infty\to \C$. We conclude that $\caT_{p,\bl}$ is continuous in its arguments and by the nuclear theorem defines a distribution in $\caD'(\C_\bx^l)$.

The fact that the derivative \eqref{actiooo} can be expressed in terms of an $n$-linear functional $\caT_n$ then comes from the fact that we take $n$-derivatives, so that the result is a sum of terms of the form
\begin{align*}
\prod_j \int \mathcal{T}_{p_j,\bl_j}(z_j,f\hdots,f) dv_g(z_j) \prod_k \mathcal{T}_{p,\bl_k}(x_{i_k},f,\hdots,f)\,,
\end{align*}
with $\sum_j |\bl_j| + \sum_k |\bl_k| =n$, where $|\bl|=\sum_i l_i$. From the definitions of $\caT_{p,\bl}$ and $m^i_{l_i}$ it follows that such integrals can be written in terms of an $n$-linear functional on $C^\infty_0(\C_\bx,M_2)$ and in the end we define  $\caT_n$ to be the resulting sum of such functionals.


For the third claim let $f_i\in C_0^\infty(\C_\bx,M_2)$, $i=1,\dots,n$ have disjoint supports and consider the perturbed metric
\begin{align*}
g^{\mu\nu}=\hat g^{\mu\nu}+\sum_{i=1}^n\epsilon_{i}f_{i}^{\mu\nu}.
\end{align*}
Next we explain that from the definition of $\mathcal{T}_n $ and the fact that the $f_i$'s have disjoint support we get
\begin{align}\label{Tneps}
\caT_{n}(f_{1},\dots,f_{n})
=(4\pi)^{n}\prod_{j=1}^{n} \partial_{\epsilon_j}|_{\epsilon_j=0} \langle\prod_{i=1}^{N}V_{\alpha_{i},g}(x_{i})\rangle_{g}.
\end{align}
To see this note that in the previous computation with \eqref{metricvar} for example the term $\partial^2_\epsilon \mu|_{\epsilon=0}(z)$ was a homogeneous polynomial of degree $2$ in the variables $\{f^{\alpha \beta} (z)\}$. 
If we were to do the same computation with the perturbation $\epsilon_1 f_1 + \epsilon_2 f_2$ of the metric and then compute $\partial_{\epsilon_1}|_0 \partial_{\epsilon_2}|_0 \mu(z)$ the result vanishes if $f_1$ and $f_2$ have disjoint support. On the other hand, the term $(\partial_\epsilon \mu|_0) (\partial_\epsilon \mu|_0)$ (meaning we hit different terms with the derivatives) is a $2$-linear functional $\mathcal{S}$ on $C^\infty_0(\C_\bx,M_2)$ so can be written as $\mathcal{S}(f,f)$ and in this case we have $\mathcal{S}(f_1,f_2)=(\partial_{\epsilon_1}\mu|_0) (\partial_{\epsilon_2}\mu|_0)$. Applying the same logic for all the terms appearing in the computation of the derivatives we end up with \eqref{Tneps}.

From the assumption that the $f_i$'s have disjoint supports it follows (as in \eqref{zeta}) that
\begin{align}\label{zeta_n}
e^\sigma \zeta &= \hat g \sum_{k=1}^\infty  \sum_{i=1}^n  (-\epsilon_i \hat g f_i^{-1})^k.
\end{align}
In particular, this implies $\partial_{\epsilon_i}\partial_{\epsilon_j}\mu(z)|_{(\epsilon_1,\hdots,\epsilon_n)=0}=0$ for $i\neq j$ so that $l_i=1$ in \eqref{Tpl} for all $i$. Furthermore 
\begin{align*}
\partial_{\epsilon_i}\mu |_{(\epsilon_1,\hdots,\epsilon_n)=0} = \partial_{\epsilon_i}\gamma_{\bar z\bar z}|_{(\epsilon_1,\hdots,\epsilon_n)=0}= \partial_{\epsilon_i}\zeta_{\bar z\bar z}|_{(\epsilon_1,\hdots,\epsilon_n)=0}=-\tfrac{1}{4}e^\sigma f_i^{zz}.
\end{align*}
Hence 
\begin{align*}
\prod_{i=1}^{k}\partial_{\epsilon_i} u_{k-1}(z) |_{(\epsilon_1,\hdots,\epsilon_n)=0}& = \left(- \tfrac{1}{4}\right)^{k }\sum_{\pi\in S_{k}} \mathcal{C}e^\sigma f^{zz}_{{\pi(1)}}\mathcal{B} e^\sigma f^{zz}_{{\pi(2)}} \mathcal{B}\dots \mathcal{B}e^\sigma f^{zz}_{{\pi(k)}}(z)\\&=\int u_k(\by,z)\prod_{j=1}^k f^{zz}_{j} (y_j) \, dv_g(y_j) 
\end{align*}
where
\begin{align*}
u_k(\by,z)=\tfrac{1}{\pi}\left(\tfrac{1}{4\pi}\right)^{k} \sum_{\pi \in S_{k}}\frac{1}{z-y_{\pi(1)}}\frac{1}{(y_{\pi(1)}-y_{\pi(2)})^2}\dots\frac{1}{(y_{\pi(k-1)}-y_{\pi(k)})^2}.
\end{align*}
From \eqref{varphi} we also get terms which contain the derivatives of $\zeta$
\begin{align*}
\partial_{\epsilon_i} |_{\epsilon_i=0} \hf\ln((1+4\zeta_{z\bar z})^{2}-4\zeta_{z z}\zeta_{\bar z\bar z}) &= \tfrac{1}{4} e^\sigma f_i^{z \bar z}\,.
\end{align*}
Thus the derivatives of $\zeta$ in the anomaly term in \eqref{fundaid} give contributions of the form.
\begin{align*}
\tfrac{c/2}{96\pi} \int R_{\hat g}(z) f_i(z)^{z \bar z} dv_{\hat g(z)}\,.
\end{align*}
The term $e^{-\sum_i \Delta_{\alpha_i} \varphi (x_i)}$ in \eqref{fundaid} contributes only derivatives of $u$ since $\zeta(x_i)=0$ by the assumption on the supports of the $f_i$'s.

We conclude that the functions $F^{\hat g}_{\mu_1 \nu_1 \hdots \mu_n \nu_n}$ in \eqref{tandf} 
are polynomials in the variables
\begin{align*}
(z_i-z_j)^{-1}, \quad (z_i-x_j)^{-1}, \quad R_{\hat g}(z_i)\,,
\end{align*}
and their complex conjugates. The complex conjugates come from the derivatives of $\ln(|1+\partial_{z}u|^{2}-|\partial_{\bar z}u|^{2})$ in \eqref{varphi} since $|1+\partial_z u|^2 = (1+\partial_z u)(1+\partial_{\bar z} \bar u)$ and the derivatives of $A(\varphi,\psi^* \hat g)$ since $(\psi^* \hat g)^{\alpha \beta}$ contains elements of the derivative matrix
\begin{align*}
D \psi &= \begin{pmatrix}
1+\partial_z u & \partial_{\bar z} u \\ \partial_z \bar u & 1 + \partial_{\bar z} \bar u
\end{pmatrix}\,.
\end{align*}
Thus the functions $F^{\hat g}_{\mu_1 \nu_1 \hdots \mu_n \nu_n}$ are smooth on the claimed region.
\end{proof}

To avoid confusion with the notations in the following computations, we now carefully explain what the previous result says about the functional derivatives of the LCFT correlation function. Recall that we are denoting $g^{\mu \nu} = \hat g^{\mu \nu} + \sum_{i=1}^n \eps_i f_i^{\mu \nu}$ and $F(\hat g,\eps_1,\hdots,\eps_n) = \langle\prod_{i=1}^{N}V_{\alpha_{i}, g}(x_{i})\rangle_{ g}$. First of all, we want to interpret \eqref{tandf} using the functional derivative from Definition \ref{functional_derivative}. Note that we can think of $\hat g \mapsto F(\hat g,0)$ as being a function of the four smooth functions given by the components of the inverse metric $\hat g^{\mu \nu}$. Denote $g_{\eps_1}^{\mu \nu} = \hat g^{\mu \nu} +  \eps_1 f_k^{\mu \nu}$. Now, provided that the functional derivatives $\frac{\delta F(\cdot,0)}{\delta \hat g^{\mu \nu}}$ exist, we have
\begin{align*}
\partial_{\eps_1}|_0 F(\hat g,\eps_1) &= \sum_{\mu, \nu} \int  \frac{\delta F(\hat g ,0)}{\delta \hat g^{\mu \nu}(z)} f_1^{\mu \nu} (z) d^2z\,,
\end{align*}
The right-hand side of the above equation has four terms, which arise from the fact that on the left-hand side $F(\hat g,\eps_1)$ is a function of all the four components $g_{\eps_1}^{\mu \nu}$, and the derivative $\partial_{\eps_1}$ has to operate on all these four arguments. By taking more derivatives we get
\begin{align*}
\prod_{k=1}^n \partial_{\eps_k}|_{\eps_k=0}F(\hat g, \eps_1,\hdots,\eps_n) &=   \int  \frac{\delta^n F(\hat g,0) }{ \delta \hat g^{\mu_1 \nu_1} (z_1) \hdots  \delta \hat g^{\mu_n \nu_n}(z_n)} \prod_{k=1}^n  f_k^{\mu_k \nu_k}(z_k) d^2z_k \,.
\end{align*}
Now comparing with \eqref{tandf} and noting that
$$
\prod_{k=1}^n \partial_{\eps_k} F(\hat g, \eps_1, \hdots, \eps_n) |_{(\eps_1,\hdots,\eps_n)=0} = \prod_{k=1}^n \partial_{\eps_k}|_{\eps_k=0}F(\hat g, \eps_1,\hdots,\eps_n)
$$
we get
\begin{align}
 e^{\sum_{k=1}^n \sigma(z_k)} F^{\hat g}_{\mu_{1} \nu_1 \dots \mu_n \nu_{n}}(z_{1},\dots,z_{n})   &=   \frac{\delta^n F(\hat g, 0) }{ \delta \hat g^{\mu_1 \nu_1} (z_1) \hdots  \delta \hat g^{\mu_n \nu_n}(z_n)},
\end{align}
where the factor $e^{\sum_k \sigma(z_k)}$ comes from the fact that in \eqref{tandf} the right-hand side contains the volume forms $dv_{\hat g}(z_k) = e^{\sigma(z_k)} d^2 z_k$. This shows that Proposition \ref{smoothness} gives the existence of the functional derivatives of the LCFT correlation functions with respect to the metric and that these functional derivatives are given in terms of the functions $F^{\hat g}_{\mu_1 \nu_1 \hdots \mu_n \nu_n}$. Thus, we introduce the notation
\begin{align}\label{correfuns}
\langle\prod_{k=1}^{n} T_{\mu_{k}\nu_{k}}(z_{k})\prod_{i=1}^N V_{\alpha_{i}, \hat g}(x_{i})\rangle_{\hat g} &:= (4 \pi)^{n} F^{\hat g}_{\mu_{1} \nu_1 \hdots \mu_n \nu_{n}}(z_{1},\dots,z_{n})\,.
\end{align}
We add the factor of $(4\pi)^{n}$ to match conventions of physics literature. This means that the left-hand side on the above expression is not the Liouville expectation of some function, but rather it is just a function of the $z_k$'s and $x_i$'s.

At this point we remark that Proposition \ref{smoothness} applies to an arbitrary metric $\tilde g$ on the Riemann sphere $\rs$, although in the proof we considered a diagonal metric $\hat g = e^\sigma \delta$. Indeed, there exists a smooth diffemorphism $\tilde \psi $ of $\rs$ such that $\tilde g = \tilde \psi^*( e^{\tilde \sigma} \delta)$ where $\tilde \sigma: \rs \to \R$ is a smooth function. Now a perturbed metric $\tilde g_\eps^{\mu \nu} = \tilde g^{\mu \nu} + \eps f^{\mu \nu}$ can be written in the form
\begin{align*}
\tilde g^{\mu \nu} + \eps f^{\mu \nu} &= ((\tilde \psi)^{-1})^* (e^{-\tilde \sigma} \delta + \eps \tilde \psi^* f)^{\mu \nu}\,.
\end{align*}
After using the diffeomorphism transformation law from Proposition \ref{correlationf} the more general result follows. Then, using \eqref{tandf} for the metric $g_{\eps_1}^{\mu \nu} = \hat g^{\mu \nu} + \eps_1 f_1^{\mu \nu}$ instead of $\hat g$ we get
\begin{align}
\prod_{j=1}^n \partial_{\epsilon_j}|_{\eps_j=0}  F(\hat g, \eps_1,\hdots,\eps_n)  &=  \partial_{\eps_1}|_{\eps_1=0} \int F^{g_{\eps_1}}_{\mu_2 \nu_2 \hdots \mu_n \nu_n}(z_2,\hdots,z_n) \prod_{j=2}^n f_j^{\mu_j \nu_j}(z_j) dv_{g_{\eps_1}}(z_j)\,.
\end{align}
Now Definition \ref{functional_derivative} applied to the derivative $\partial_{\eps_1}|_{\eps_1=0}$ on the right-hand side implies
\begin{align}
&\prod_{j=1}^n \partial_{\epsilon_j}  F(\hat g, \eps_1,\hdots,\eps_n) |_{(\epsilon_1,\hdots,\epsilon_n)=0} \nonumber \\
&= \int f_1^{\mu_1 \nu_1}(z_1) \frac{\delta}{\delta \hat g^{\mu_1 \nu_1}(z_1)} F^{\hat g}_{\mu_2 \nu_2 \hdots \mu_n \nu_n}(z_2,\hdots,z_n) \prod_{j=2}^n f_j^{\mu_j \nu_j}(z_j) dv_{\hat g}(z_j) d^2 z_1\,,
\end{align}
where we sum over repeated indices. Comparing this with \eqref{tandf} gives
\begin{align}\label{fditer}
e^{\sigma(z_1)} F^{\hat g}_{\mu_{1} \nu_1\dots \mu_n \nu_{n}}(z_{1},\dots,z_{n}) &=  \frac{\delta}{\delta \hat g^{\mu_1 \nu_1}(z_1)} F^{\hat g}_{\mu_2 \nu_2 \hdots \mu_n \nu_n}(z_2,\hdots,z_n)\,.
\end{align}
We will use this equation later.
\begin{definition}\label{def_notation}
It is natural to define the notations
\begin{align}\label{t_lin}
\langle (a T_{\mu \nu}(z) + f(z))  \prod_{i=1}^n T_{\mu_i \nu_i}(z_i) \prod_{j=1}^N V_{\alpha_j,g}(x_j) \rangle_g &:= a \langle T_{\mu \nu}(z)  \prod_{i=1}^n T_{\mu_i \nu_i}(z_i) \prod_{j=1}^N  V_{\alpha_j,g}(x_j) \rangle_g \\
& \quad + f(z) \langle  \prod_{i=1}^n T_{\mu_i \nu_i}(z_i)  \prod_{j=1}^N V_{\alpha_j,g}(x_j) \rangle_g\,, \nonumber
\end{align}
where $a \in \C$ and $f: \C \to \C$ is a smooth function. 
We will also denote
\begin{align}\label{t_notation}
\langle \partial_z T_{\mu \nu}(z) \prod_{i=1}^n T_{\mu_i \nu_i}(z_i) \prod_{j=1}^N V_{\alpha_j,g}(x_j) \rangle_g &:= \partial_z F^{g}_{\mu \nu \hdots \mu_n \nu_n}(z,z_1,\hdots,z_n)\,,
\end{align}
and
\begin{align}\label{tpsi_notation}
&\langle \sum_{\alpha, \beta } (D\psi^{T})_{\mu \alpha}(z)T_{\alpha \beta}(\psi(z)) (D\psi)_{\beta \nu}(z) \prod_{i=1}^n T_{\mu_i \nu_i}(z_i)  \prod_{j=1}^N V_{\alpha_j,g}(x_j) \rangle_g \\
&:= \sum_{\alpha, \beta} (D \psi^T)_{\mu \alpha} (z) (D \psi)_{\beta \nu} (z) F^{g}_{\alpha \beta \mu_1 \nu_1 \hdots \mu_n \nu_n}(\psi(z),z_1,\hdots,z_n) \,, \nonumber
\end{align}
where $\psi: \rs \to \rs$ is a smooth diffeomorphism. Denote
\begin{align*}
\tilde T_{\mu \nu}(z) = \sum_{\alpha,\beta} (D \psi)^T_{\mu \alpha} T_{\alpha \beta} (\psi(z)) (D \psi)_{\beta \nu}(z).
\end{align*}
The notation \eqref{tpsi_notation} has an obvious generalisation for multiple $\tilde T_{\mu_i \nu_i} + f_i$-insertions, where $f_i: \C \to \C$ is a smooth function, given recursively by
\begin{align}\label{tt_notation}
&\langle \prod_{k=1}^m \big( \tilde T_{\alpha_k \beta_k}(z_k) + f_k(z_k) \big) \prod_{i=1}^n T_{\mu_i \nu_i}(w_i) \prod_{j=1}^N V_{\alpha_j,g}(x_j) \rangle_g \\
&:= \langle \prod_{k=2}^m   \big( \tilde T_{\alpha_k \beta_k}(z_k) + f_k(z_k) \big) \tilde T_{\alpha_1 \beta_1}(z_1) \prod_{i=1}^n T_{\mu_i \nu_i}(w_i)  \prod_{j=1}^N V_{\alpha_j,g}(x_j) \rangle_g  \nonumber \\
& \quad + f_1(z_1) \langle \prod_{k=2}^m   \big( \tilde T_{\alpha_k \beta_k}(z_k) + f_k(z_k) \big) \prod_{i=1}^n T_{\mu_i \nu_i}(w_i) \prod_{j=1}^N V_{\alpha_j,g}(x_j) \rangle_g. \nonumber
\end{align}
\end{definition}
For the time being let us consider the variation of these $T$-correlation functions under diffeomorphisms and Weyl transformations. 

Let us define  
\begin{align}\label{aab_def}
a_{\alpha\beta}(\varphi,g,z):=4\pi c \frac{\delta}{\delta g(z)^{\alpha\beta}}A(\varphi,g).
\end{align}
In conformal coordinates $\hat g_{\alpha\beta}=e^\sigma \delta_{\alpha\beta}$ we have from a computation similar to Lemma \ref{variationlemma}
\begin{align}\label{aab}
a_{\alpha\beta}(\varphi,\hat g,z)=\tfrac{c}{24\pi}(\partial_\alpha\varphi(z) \partial_\beta\varphi(z)-\hf \delta_{\alpha\beta}(\partial\varphi(z))^2-\partial_\alpha\partial_\beta\varphi(z)+\hf(\partial_\alpha\sigma(z)\partial_\beta\varphi(z)+\partial_\beta\sigma(z)\partial_\alpha\varphi(z))).
\end{align}
Furthermore we have the locality property for  $z\neq z'$:
\begin{align}
\frac{\delta^{2}}{\delta \hat g^{\mu \nu}(z)\delta \hat g^{\mu' \nu'}(z')}A(\varphi,\hat g)=0\label{locality}.
\end{align}
Indeed by the definition of the functional derivative in Definition \ref{functional_derivative}
\begin{align*}
\frac{\delta^{2}}{\delta \hat g^{\mu \nu}(z)\delta \hat g^{\mu' \nu'}(z')}A(\varphi,\hat g) = \frac{1}{4 \pi c} \frac{\delta}{\delta g^{\mu'\nu'}(z')}a_{\mu \nu}(\varphi,\hat g, z) &= 0
\end{align*}
for $z' \neq z$.

Then

\begin{proposition} \label{tinsert}
Let $(x_1,\hdots,x_N) \in \C^N$ and $(z_1,\hdots,z_n) \in \C_\bx^n$ be tuples of disjoint points. Then the stress-energy tensor correlation functions satisfy
\begin{align}
\langle\prod_{i=1}^n T_{\mu_{i}\nu_{i}}(z_{i})\prod_{j=1}^N V_{\alpha_{j},\psi^* g}(x_{j})\rangle_{\psi^{\ast}g}=\langle\prod_{i=1}^n \tilde T_{\mu_{i}\nu_{i}}(z_{i})\prod_{j=1}^N V_{\alpha_{j},g}(\psi(x_{j}))\rangle_{g}\label{emdiff}
\end{align}
where $\tilde T_{\mu \nu} =  \sum_{\alpha,\beta} (D\psi^{T})_{\mu \alpha}(T_{\alpha \beta}\circ\psi) (D\psi)_{\beta \nu}$ and the right-hand side above is defined in \eqref{tt_notation}. Also,
\begin{align}
\langle\prod_{i=1}^n T_{\mu_{i}\nu_{i}}(z_{i})\prod_{j=1}^N V_{\alpha_{j},e^\varphi g}(x_{j})\rangle_{e^{\varphi}g}=e^{cA(\varphi,g)}\prod _{k=1}^N e^{-\Delta_{\alpha_k}\varphi(x_{k})}\langle\prod_{i=1}^n (T_{\mu_{i}\nu_{i}}(z_{i})+a_{\mu_{i}\nu_{i}}(\varphi,g,z_{i}))\prod_{j=1}^N V_{\alpha_{j},g}(x_{j})\rangle_{g}\label{emweyl}
\end{align}
where the right-hand side above is defined in \eqref{tt_notation} by setting $f_k(z)=a_{\mu_k \nu_k}(\varphi,g,z_k)$ and $\psi(z)=z$.

\end{proposition}

\proof

For the first claim  let  
\begin{align*}
h^{\alpha\beta} &=(\psi^\ast g)^{\alpha\beta}+\sum_{i=1}^n \epsilon_i f_i^{\alpha\beta},
\end{align*}
where the $f_i$'s have disjoint supports. Then
\begin{align*}
h=\psi^\ast \tilde g
\end{align*}
where
\begin{align*}
\tilde g^{\alpha\beta}= g^{\alpha\beta}+\sum_{i=1}^n \epsilon_i \tilde f_i^{\alpha\beta}+\gamma^{\alpha\beta}
\end{align*}
with $\tilde f=(D\psi^{-1})^T (f\circ\psi^{-1}) (D\psi^{-1})$ and $\gamma = \mathcal{O}(\eps_i \eps_j)$ meaning $\partial_{\eps_i} \gamma|_{(\eps_1,\hdots,\eps_n)=0} =0$. We denote $ F(g,x_1,\hdots,x_N)=\langle \prod_{i=1}^N V_{\alpha_i,g}(x_i)\rangle_g$. Thus by \eqref{tandf}, \eqref{correfuns} and the diffeomorphism covariance in Proposition \ref{correlationf} we get (we sum over repeated indices)
\begin{align*}
\int \langle\prod_{i=1}^n T_{\mu_{i}\nu_{i}}(z_{i}) &\prod_{j=1}^N V_{\alpha_{j},g}(x_{j})\rangle_{\psi^{\ast}g}\prod_{i=1}^n f^{\mu_{i}\nu_{i}}(z_i)dv_{\psi^\ast g}(z_i) \\
&=(4\pi)^{n}\prod_{i=1}^{n}\partial_{\epsilon_i} |_{0} F(\psi^\ast \tilde g,x_1,\hdots,x_N)\\
&=(4\pi)^{n}\prod_{i=1}^{n}\partial_{\epsilon_i} |_{0} F( \tilde g,\psi(x_1),\hdots,\psi(x_N))\\
&=\int \langle\prod_{i=1}^n T_{\mu_{i}\nu_{i}}(z_{i})\prod_{j=1}^N V_{\alpha_{j},g}(\psi(x_{j}))\rangle_{g}\prod_{i=1}^n \tilde f^{\mu_{i}\nu_{i}}(z_i)dv_g(z_i)\\
&=\int \langle\prod_{i=1}^n \tilde T_{\mu_{i}\nu_{i}}(z_{i})\prod_{j=1}^N V_{\alpha_{j},g}(\psi(x_{j}))\rangle_{g}\prod_{i=1}^n f^{\mu_{i}\nu_{i}}(z_i)dv_{\psi^\ast g}(z_i).
\end{align*}
In the last equality we did a change of variables $z_i \to \psi(z_i)$ and used 
\begin{align*}
\sum_{\mu,\nu} (T_{\mu \nu} \circ \psi)  (\tilde f^{\mu \nu} \circ \psi)  &= \sum_{\mu,\nu} \sum_{\alpha,\beta} (D \psi^T)_{\beta \nu } (T_{\mu \nu} \circ \psi) (D \psi)_{\mu \alpha} f^{\alpha \beta} =\sum_{\alpha,\beta } \tilde T_{\beta \alpha} f^{\beta \alpha}\,,
\end{align*}
where we used $D \psi^{-1}(\psi(z)) = (D \psi^T)(z)$ and the symmetry of $T$ and $f$. Now the first claim is proven.

The second claim follows from the Weyl transformation law in Proposition \ref{correlationf} in the following way.  Let 
$$h^{\alpha\beta}:=e^{-\varphi}g^{\alpha\beta}+\sum_{i=1}^n \eps_i f_i^{\alpha\beta}$$
and
$$\tilde h^{\alpha\beta}:=g^{\alpha\beta}+\sum_{i=1}^n\eps_ie^\varphi f_i^{\alpha\beta}.$$
Now $h = e^{\varphi} \tilde h$. Denote $F(g) = \langle \prod_{j=1}^N V_{\alpha_j, g}(x_j) \rangle_{g}$. Then
\begin{align}\label{tweyl1}
\int\langle\prod_{i=1}^n T_{\mu_{i}\nu_{i}}(z_{i})\prod_{j=1}^N V_{\alpha_{j},e^\varphi g}(x_{j})\rangle_{e^{\varphi}g}\prod_{i=1}^n f^{\mu_{i}\nu_{i}}(z_i)dv_{e^{\varphi}g} &=(4\pi)^{n}\prod_{i=1}^{n} \partial_{\epsilon_i} |_{\epsilon_i=0} F(h)\\
&=(4\pi)^{n}\prod_{j=1}^N e^{-\Delta_{\alpha_j}\varphi(x_{j})}\prod_{i=1}^{n} \partial_{\epsilon_i} |_{0}\big( e^{cA(\varphi,\tilde h)} F(\tilde h)\big), \nonumber
\end{align}
where the first equality comes from \eqref{tandf} and \eqref{correfuns}, and the second equality uses the Weyl transformation law from Proposition \ref{correlationf}. To compute the derivative on the last line we apply the locality \eqref{locality}. Indeed, it implies
\begin{align}\label{Aloc}
\partial_{\epsilon_1} \partial_{\epsilon_2} A(\varphi,\tilde h) |_{(\epsilon_1,\epsilon_2)=0} &= \int f^{\mu \nu}(z) f^{\mu' \nu'}(z') \frac{\delta^2 A(\varphi,\tilde h)}{\delta g^{\mu \nu}(z) \delta g^{\mu' \nu'}(z')} \, d^2z\, d^2z'  = 0\,,
\end{align}
where we sum over repeated indices. Now differentiating the product and using \eqref{Aloc} gives
\begin{align*}
\prod_{i=1}^{n} \partial_{\epsilon_i} |_{0}\big( e^{cA(\varphi,\tilde h)} F(\tilde h)\big) &= \prod_{i=2}^n \partial_{\eps_i}|_0 \big( \partial_{\eps_1} (e^{A(\varphi, \tilde h)}) F(\tilde h) + e^{A(\varphi, \tilde h)} \partial_{\eps_1} F(\tilde h) \big)|_{\eps_1=0} \\
&= \prod_{i=2}^{n} \partial_{\eps_i}|_0 \big(  e^{c A(\varphi,\tilde h)}(   \partial_{\eps_1} + c \partial_{\eps_1} A(\varphi, \tilde h) )  F(\tilde h)  \big)|_{\eps_1=0} \\
&= e^{cA(\varphi,g)} \prod_{i=1}^n   ( \partial_{\eps_i} + c \partial_{\eps_i} A(\varphi, \tilde h)) F(\tilde h) |_{\eps_i=0} \,.
\end{align*}
Using \eqref{tandf} and \eqref{correfuns} again gives
\begin{align*}
\prod_{i=1}^n   (\partial_{\eps_i} + c \partial_{\eps_i} A(\varphi, \tilde h))|_{\eps_i=0} F(\tilde h) &= \tfrac{1}{(4\pi)^n} \int \prod_{i=1}^n e^{\varphi(z_i)} f_i^{\mu_i \nu_i}(z_i) \langle (T_{\mu_i \nu_i}(z_i) + a_{\mu_i \nu_i}(\varphi,g,z_i)) \prod_{j=1}^N V_{\alpha_j, g}(x_j) \rangle_g.
\end{align*}

Plugging the result back into \eqref{tweyl1} yields
\begin{align*}
&\int\langle\prod_{i=1}^n T_{\mu_{i}\nu_{i}}(z_{i})\prod_{j=1}^N V_{\alpha_{j},e^\varphi g}(x_{j})\rangle_{e^{\varphi}g}\prod_{i=1}^n f^{\mu_{i}\nu_{i}}(z_i)dv_{e^{\varphi}g}\\
&= \prod_{j=1}^N e^{-\Delta_{\alpha_j}\varphi(x_{j})}e^{cA(\varphi,g)}\int \langle\prod_{i=1}^n(T_{\mu_{i}\nu_{i}}(z_{i})+a_{\mu_{i}\nu_{i}}(\varphi,g,z_i))\prod_{j=1}^N V_{\alpha_{j},g}(x_{j})\rangle_{g}\prod_{i=1}^{n}f^{\mu_{i}\nu_{i}}(z_i) e^{\varphi(z_i)} dv_g(z_i).
\end{align*} 
The claim now follows from $e^\varphi v_g=v_{e^\varphi g}$. \qed

\vskip 2mm

\subsection{Ward identities}\label{1st_ward_section}

In this section we will demonstrate that Propositions \ref{smoothness} and \ref{tinsert} allow us to compute $T$-correlations  inductively.  By Proposition \ref{tinsert} it suffices to do this in the conformal coordinates. Furthermore, only $T_{zz}$ and  $T_{\bar z\bar z}$ correlations are non-trivial and they may be computed separately (see Remark \ref{Tbarz} below). 
We have then
\begin{proposition} The Ward identity \eqref{wardid} holds whenever $z_i\neq z_j$ for $i \neq j$, and $z_i \neq x_j$ for all $i$ and $j$. \label{ward_prop}
\end{proposition}

\proof
Let $\hat g=e^\sigma \delta$, where $\delta$ denotes the Euclidean metric, and consider the perturbed metric $g$ with the inverse metric having the components
$$g^{z\bar z}=\hat g^{z\bar z}= 2 e^{-\sigma},\ \ \  g^{zz}=\epsilon f,$$
where  $f$ is real and has compact support outside the points $\{z_i\}_{i=2}^n$ and $\{x_j\}_{j=1}^N$, and  $\epsilon\in\C$ (so that  $g^{\bar z\bar z}= \bar \epsilon f$). We will denote $\eps=\eps_1+i\eps_2$, $\eps_1,\eps_2 \in \R$ and $\partial_\eps = \tfrac{1}{2}(\partial_{\eps_1}-i \partial_{\eps_2})$ so that $\partial_\eps \bar \eps = 0$. By \eqref{correfuns} and \eqref{fditer}
\begin{align*}
\int  f(z_{1})\langle T_{zz}(z_{1})\prod_{i=2}^{n} T_{zz}(z_{i})\prod_{j=1}^N V_{\alpha_{j}, \hat g}(x_{j})\rangle_{\hat g}dv_{\hat g}(z_{1}) &= (4\pi)^n \int f(z_1) F^{\hat g}_{zz \hdots zz}(z_1,\hdots,z_n) dv_{\hat g}(z_1) \\
&= (4\pi)^n \int f(z_1)  \frac{\delta}{\delta \hat g^{zz}(z_1)} F^{\hat g}_{zz \hdots zz}(z_2,\hdots,z_n)  d^2 z_1\,.
\end{align*}
Next we apply Definition \ref{functional_derivative} and \eqref{correfuns}
\begin{align*}
(4 \pi)^n \int f(z_1)  \frac{\delta}{\delta \hat g^{zz}(z_1)} F^{\hat g}_{zz \hdots zz}(z_2,\hdots,z_n)  d^2 z_1&= (4 \pi)^n \partial_\eps|_0 F^{g}_{zz \hdots zz}(z_2,\hdots,z_n)   \\
&= 4 \pi \partial_\eps|_0 \langle \prod_{i=2}^n T_{zz}(z_i) \prod_{j=1}^N V_{\alpha_j, g}(x_j) \rangle_g\,.
\end{align*}
To compute the $\eps$-derivative we want to utilize the Weyl and diffeomorphism transformation laws from Proposition \ref{tinsert}. To this end we write (see Section \ref{beltrami_section})
\begin{align*}
g=e^{\varphi}\psi^{\ast}\hat g.
\end{align*}
Note that the dependence on $\eps$ is in $\varphi$ and $\psi$.
Using the Weyl transformation law from Proposition \ref{tinsert} we get
\begin{align*}
&\partial_{\epsilon}\big|_{0}\langle\prod_{i=2}^{n} T_{zz}(z_{i})\prod_{j=1}^N V_{\alpha_{j},g}(x_{j})\rangle_{g} \\
&= \partial_{\epsilon}\big|_{0}\langle\prod_{i=2}^{n} T_{zz}(z_{i})\prod_{j=1}^N V_{\alpha_{j},e^{\varphi} \psi^* \hat g}(x_{j})\rangle_{e^\varphi \psi^* \hat g} \\
&= \partial_{\eps}\big|_0 \Big( e^{c A(\varphi, \psi^* \hat g)} \prod_{k=2}^N e^{- \Delta_{\alpha_k} \varphi(x_k)} \langle \prod_{i=2}^n \big( T_{\mu_i \nu_i}(z_i) + a_{\mu_i \nu_i}(\varphi, \psi^* \hat g, z_i) \big) \prod_{j=1}^N V_{\alpha_j, \psi^* \hat g}(x_j) \rangle_{\psi^* \hat g} \Big)\,.
\end{align*}
Note that the last line uses the notation introduced in Definition \ref{def_notation}. We can further simplify this by using the diffeomorphism transformation law from Proposition \ref{tinsert} (and again using notations from Definition \ref{def_notation})
\begin{align*}
&\langle \prod_{i=2}^n \big( T_{\mu_i \nu_i}(z_i) + a_{\mu_i \nu_i}(\varphi, \psi^* \hat g, z_i) \big) \prod_{j=1}^N V_{\alpha_j, \psi^* \hat g}(x_j) \rangle_{\psi^* \hat g} \\
&= \langle \prod_{i=2}^n \big( \tilde T_{\mu_i \nu_i}(z_i) + a_{\mu_i \nu_i}(\varphi, \psi^* \hat g, z_i) \big) \prod_{j=1}^N V_{\alpha_j, \hat g}(\psi(x_j)) \rangle_{\hat g}\,.
\end{align*}
We arrive at
\begin{align}\label{com_step0}
&\partial_{\epsilon}\big|_{0}\langle\prod_{i=2}^{n} T_{zz}(z_{i})\prod_{j=1}^N V_{\alpha_{j},g}(x_{j})\rangle_{g} \\
&= \partial_{\epsilon}\big|_{0}\Big( e^{c A(\varphi, \psi^* \hat g)} \prod_{k=1}^N e^{- \Delta_{\alpha_k} \varphi(x_k)} \langle \prod_{i=2}^n \big( \tilde T_{zz}(z_i) + a_{zz}(\varphi, \psi^* \hat g, z_i) \big) \prod_{j=1}^N V_{\alpha_j, \hat g}(\psi(x_j)) \rangle_{\hat g} \Big). \nonumber
\end{align}
To proceed with the computation we now compute $\psi$ and $\varphi$ to first order in $\eps$. First, recall equations \eqref{geepsi} and \eqref{muepsi}. To first order in $\epsilon$ the metric is $g_{\bar z\bar z}=-\frac{\epsilon}{4} e^{2\sigma}f$. Hence 
\begin{align*}
\zeta_{\bar z\bar z}=-\tfrac{\epsilon}{4}e^{\sigma}f+\caO(\eps^2)
\end{align*}
 and $\zeta_{ z\bar z}=0$. Then, for the $\gamma$ appearing in \eqref{belcoef} we have  $\gamma_{\bar z\bar z}=-\frac{\epsilon}{4}e^{\sigma}f + \mathcal{O}(\eps^2)$ so that
 \begin{align*}
 \mu=-\frac{\epsilon}{4}e^{\sigma}f + \mathcal{O}(\eps^2).
\end{align*}
Now recall from Section \ref{beltrami_section}
\begin{align*}
\psi(z) &= z + \sum_{n=0}^\infty (\mu \mathcal{B})^n \mu\,.
\end{align*}
Thus
\begin{align}\label{uuepsil}
\partial_{\epsilon}\big|_{0}\psi(z)=-\tfrac{1}{4}\caC(e^{\sigma}f)(z)=-\tfrac{1}{4\pi}\int\frac{1}{z-x}f(x)dv_{\hat g}(x)=:u(z).
\end{align}
From \eqref{varphi} we infer
\begin{align}\label{varhiepsil}
\partial_{\epsilon}|_{0}\varphi=
-u\partial_{z}\sigma-\partial_{z}u.
\end{align}
Based on these formulas we remark that 
\begin{align*}
\psi(z)|_{\eps=0} = z\,, \quad \varphi(z)|_{\eps=0} = 0\,.
\end{align*}
The latter equation together with \eqref{aab} further imply $a_{zz}(\varphi,\psi^* \hat g,z)|_{\eps=0}=0$. Thus, continuing from \eqref{com_step0} we get
\begin{align}\label{com_step1}
& \partial_{\epsilon}\big|_{0}\Big( e^{c A(\varphi, \psi^* \hat g)} \prod_{k=1}^N e^{- \Delta_{\alpha_k} \varphi(x_k)} \langle \prod_{i=2}^n \big( \tilde T_{zz}(z_i) + a_{zz}(\varphi, \psi^* \hat g, z_i) \big) \prod_{j=1}^N V_{\alpha_j, \hat g}(\psi(x_j)) \rangle_{\hat g} \Big) \\
&= c \partial_\eps|_0 A(\varphi, \psi^* \hat g) \langle\prod_{i=2}^{n} T_{zz}(z_{i})\prod_{j=1}^N V_{\alpha_{j},\hat g}(x_{j})\rangle_{\hat g} - \sum_{k=2}^n \Delta_{\alpha_k} \partial_{\eps}|_0 \varphi(x_k) \langle\prod_{i=2}^{n} T_{zz}(z_{i})\prod_{j=1}^N V_{\alpha_{j},\hat g}(x_{j})\rangle_{\hat g} \nonumber \\
&+ \partial_{\epsilon}\big|_{0}  \langle \prod_{i=2}^n \big( \tilde T_{zz}(z_i) + a_{zz}(\varphi, \psi^* \hat g, z_i) \big) \prod_{j=1}^N V_{\alpha_j, \hat g}(x_j) \rangle_{\hat g} +\partial_\epsilon\big|_0  \langle \prod_{i=2}^n  T_{zz}(z_i)  \prod_{j=1}^N V_{\alpha_j, \hat g}(\psi(x_j)) \rangle_{\hat g}. \nonumber
\end{align}
Next we evaluate all the derivatives appearing above. Since $\varphi=(-u \partial_z \sigma - \partial_z u) \eps + \mathcal{O}(\eps^2)$ and $ e^{\sigma}R_{ e^{\sigma} \delta}=-4\partial_{z}\partial_{\bar z}\sigma$ we get  
\begin{align}\label{Beq}
c \partial_{\epsilon}\big|_{0} A(\varphi,\psi^{\ast} \hat g) 
&=\tfrac{c}{48\pi}\int  R_{ e^{\sigma}\delta}\partial_{\epsilon}\big|_{0}\varphi dv_{ e^{\sigma}\delta}(z)
=\tfrac{c}{12\pi}\int \partial_{z}\partial_{\bar z}\sigma(u\partial_{z}\sigma+\partial_{z}u)d^{2}z\\
&=\tfrac{c}{12\pi}\int u(\partial_{\bar z}\partial_{z }\sigma\partial_{z}\sigma- \partial_{\bar z}\partial_{z }^{2}\sigma)d^{2}z=\tfrac{c}{12\pi}\int \partial_{\bar z}u(\partial_{z }^2\sigma-\hf (\partial_{z }\sigma)^{2})d^{2}z \nonumber \\&=-\tfrac{c}{48\pi}\int f( \partial_{z }^2\sigma - \hf (\partial_{z }\sigma)^{2})dv_{\hat g}=-\tfrac{c}{48\pi}\int ftdv_{\hat g} \nonumber
\end{align}
where $t$ is defined in \eqref{tcorre}.

The second term on the right-hand side of \eqref{com_step1} is readily evaluated by using \eqref{varhiepsil}
\begin{align*}
- \sum_{k=2}^n \Delta_{\alpha_k} \partial_\eps|_0 \varphi(x_k) &= \sum_{k=2}^n \Delta_{\alpha_k} (u(z_k) \partial_{z_k} \sigma(z_k) + \partial_{z_k} u(z_k) ) \\
&= \tfrac{1}{4\pi} \sum_{k=2}^n \Delta_{\alpha_k} \int  f(z) \Big(  \frac{\partial_{z_k} \sigma(z_k)}{z-z_k} + \frac{1}{(z-z_k)^2} \Big) dv_{\hat g}(z).
\end{align*}

For the third term on the right-hand side of \eqref{com_step1} we have
\begin{align*}
&\partial_{\epsilon}\big|_{0}  \langle \prod_{i=2}^n \big( \tilde T_{zz}(z_i) + a_{zz}(\varphi, \psi^* \hat g, z_i) \big) \prod_{j=1}^N V_{\alpha_j, \hat g}(x_j) \rangle_{\hat g} \\
&= \partial_\eps \big|_{\eps=0}  \langle \prod_{i=2}^n  \tilde T_{zz}(z_i)  \prod_{j=1}^N V_{\alpha_j, \hat g}(x_j) \rangle_{\hat g} \\
& \quad + \sum_{k=2}^n \partial_\eps |_0 a_{zz}(\varphi,\psi^* \hat g, z_k) \langle \prod_{i \neq 1,k}^n  T_{zz}(z_i) \prod_{j=1}^N V_{\alpha_j, \hat g}(x_j) \rangle_{\hat g},
\end{align*}
where we again used $a_{zz}(\varphi, \psi^* \hat g, z_i)|_{\eps=0} = 0$. Denote $b(z_{j},f)=\partial_{\epsilon}\big|_{0}a_{zz}(\varphi,\psi^{\ast} \hat g,z_{j})$. Using $\varphi=\caO(\epsilon)$, \eqref{aab_def} and \eqref{varhiepsil} we get
\begin{align*}
b(x,f)&=-\tfrac{c}{12}\frac{\delta}{\delta g^{zz}(x)}\int  (u(z)\partial_{z}\sigma(z)+\partial_{z}u(z))R_{g}(z)dv_{g}(z)|_{g=e^{\sigma}\delta}.
\end{align*}
The derivative is calculated in  Lemma \ref{variationlemma}:
\begin{align*}
b(x,f)&=-\tfrac{c}{12} \big(-\partial_x^2+\partial_x\sigma(x) \partial_x \big)  \big( u(x)\partial_{x}\sigma(x)+\partial_{x}u(x) \big).
\end{align*}
Some algebra then gives 
\begin{align*}
b(x,f)&=\tfrac{c}{12} \big(\partial_x^3u(x)-2t(x)\partial_x u(x)-u(x)\partial_x t(x) \big)\\
&=\tfrac{c}{48\pi}\int  f(z) \Big(\frac{6}{(z-x)^4}+\frac{2t(x)}{(z-x)^2}+\frac{\partial_x t(x)}{z-x} \Big)dv_{\hat g}(z).
\end{align*}

Next we compute (recalling Definition \ref{def_notation})
\begin{align*}
&\partial_\eps \big|_{\eps=0}  \langle \prod_{i=2}^n  \tilde T_{zz}(z_i)  \prod_{j=1}^N V_{\alpha_j, \hat g}(x_j) \rangle_{\hat g} \\
&= \sum_{k=2}^n \partial_\eps|_0 \Big( \sum_{\mu,\nu} (D\psi)^T_{z \mu}(z_k) (D\psi)_{\nu z}(z_k) \langle T_{\mu \nu}(\psi(z_k)) \prod_{i\neq 1,k}^n T_{zz}(z_i) \prod_{j=1}^N V_{\alpha_{j}, \hat g}(x_{j})\rangle_{\hat g} \Big),
\end{align*}
where we used $\psi(z)|_{\eps=0} = z$. We have
\begin{align*}
D \psi &= \begin{pmatrix}
\partial_z \psi & \partial_{\bar z} \psi \\ \partial_z \bar \psi & \partial_{\bar z} \bar \psi
\end{pmatrix} = \begin{pmatrix}
1 + \partial_z u & \partial_{\bar z} u \\ \partial_z \bar u & 1 + \partial_{\bar z} \bar u
\end{pmatrix}\,.
\end{align*}
Thus
\begin{align*}
\partial_\eps|_0 (D\psi)_{zz}(z_k) &=  \partial_{z_k} u(z_k) \,, \\
\partial_\eps|_0 (D\psi)_{z \bar z}(z_k) &=  \partial_{\bar z_k} u(z_k) = - \tfrac{1}{4} e^{\sigma(z_k)} f(z_k) = 0 \,, \\
\partial_\eps|_0 (D\psi)_{\bar z \bar z}(z_k) &= 0 \,, \\
\partial_\eps|_0 (D\psi)_{\bar z z }(z_k) &= 0\,,
\end{align*}
where in the second equality we used $\partial_{\bar z} \mathcal{C}(e^\sigma f) = e^\sigma f$ and the assumption that the support of $f$ does not intersect $\{z_2,\hdots,z_k\}$ and in the third and the fourth equalities we used $\partial_\eps \bar \eps = 0$. We get
\begin{align*}
& \sum_{k=2}^n \partial_\eps|_0 \Big( \sum_{\mu,\nu} (D\psi)^T_{z \mu}(z_k) (D\psi)_{\nu z}(z_k) \langle T_{\mu \nu}(\psi(z_k)) \prod_{i\neq 1,k}^n T_{zz}(z_i) \prod_{j=1}^N V_{\alpha_{j}, \hat g}(x_{j})\rangle_{\hat g} \Big) \\
&= \sum_{k=2}^n  \Big(\partial_\eps|_0 (D\psi)^T_{zz}(z_k) + \partial_\eps|_0 (D\psi)_{zz}(z_k) \Big) \langle \prod_{i=2}^n T_{zz}(z_i) \prod_{j=1}^N V_{\alpha_j,\hat g}(x_j) \rangle_{\hat g}  \\
& \quad + \sum_{k=2}^n \partial_\eps|_0 \psi(z_k) \partial_{z_k} \langle \prod_{i=2}^n T_{zz}(z_i) \prod_{j=1}^N V_{\alpha_j,\hat g}(x_j) \rangle_{\hat g} \\
&= \sum_{k=2}^n \big( 2 \partial_{z_k} u(z_k) + u(z_k) \partial_{z_k} \big) \langle \prod_{i=2}^n T_{zz}(z_i) \prod_{j=1}^N V_{\alpha_j,\hat g}(x_j) \rangle_{\hat g}  \\
&=\tfrac{1}{4\pi} \sum_{k=2}^n \int f(z) \Big( \frac{2}{(z-z_k)^2} + \frac{\partial_{z_k}}{z-z_k} \Big)  \langle  \prod_{i=2}^{n} T_{zz}(z_{i})\prod_j V_{\alpha_{j}, \hat g}(x_{j}) \rangle_{\hat g} dv_{\hat g}(z).
\end{align*}

By the chain rule the last term on the right-hand side of \eqref{com_step1} takes the form
\begin{align*}
\partial_\epsilon\big|_0  \langle \prod_{i=2}^n  T_{zz}(z_i)  \prod_{j=1}^N V_{\alpha_j, \hat g}(\psi(x_j)) \rangle_{\hat g} &= \sum_{l=1}^N \partial_\eps|_0 \psi(x_l) \partial_{x_l} \langle \prod_{i=2}^n  T_{zz}(z_i)  \prod_{j=1}^N V_{\alpha_j, \hat g}(x_j) \rangle_{\hat g} \\
&= \sum_{l=1}^N u(x_l) \partial_{x_l} \langle \prod_{i=2}^n  T_{zz}(z_i)  \prod_{j=1}^N V_{\alpha_j, \hat g}(x_j) \rangle_{\hat g} \\
&= \tfrac{1}{4\pi} \sum_{l=1}^N \int f(z) \frac{1}{z-x_l} \partial_{x_l} \langle \prod_{i=2}^n  T_{zz}(z_i)  \prod_{j=1}^N V_{\alpha_j, \hat g}(x_j) \rangle_{\hat g} dv_{\hat g}(z).
\end{align*}


Collecting all the derivatives we computed together we get
\begin{align*}
&\partial_{\epsilon}\big|_{0}\langle\prod_{i=2}^{n} T_{zz}(z_{i})\prod_{j=1}^N V_{\alpha_{j},g}(x_{j})\rangle_{g} \\
&= -\tfrac{c}{48\pi}\int f(z) t(z) dv_{\hat g}(z)  \langle\prod_{i=2}^{n} T_{zz}(z_{i})\prod_{j=1}^N V_{\alpha_{j},\hat g}(x_{j})\rangle_{\hat g} \\
& \quad + \tfrac{1}{4\pi} \sum_{k=2}^n \Delta_{\alpha_k} \int  f(z) \Big(  \frac{\partial_{z_k} \sigma(z_k)}{z-z_k} + \frac{1}{(z-z_k)^2} \Big) dv_{\hat g}(z) \langle\prod_{i=2}^{n} T_{zz}(z_{i})\prod_{j=1}^N V_{\alpha_{j},\hat g}(x_{j})\rangle_{\hat g} \\
& \quad + \tfrac{c}{48\pi} \sum_{k=2}^n \int \Big(\frac{6}{(z-z_k)^4}+\frac{2t(z_k)}{(z-z_k)^2}+\frac{\partial_z t(z_k)}{z-z_k} \Big)f(z)dv_{\hat g}(z) \langle\prod_{i \neq 1,k}^{n} T_{zz}(z_{i})\prod_{j=1}^N V_{\alpha_{j},\hat g}(x_{j})\rangle_{\hat g} \\
& \quad + \tfrac{1}{4\pi} \sum_{k=2}^n \int f(z) \Big( \frac{2}{(z-z_k)^2} + \frac{\partial_{z_k}}{z-z_k} \Big)  \langle  \prod_{i=2}^{n} T_{zz}(z_{i})\prod_j V_{\alpha_{j}, \hat g}(x_{j}) \rangle_{\hat g} dv_{\hat g}(z) \\
& \quad + \tfrac{1}{4\pi} \sum_{l=1}^N \int f(z) \frac{1}{z-x_l} \partial_{x_l} \langle \prod_{i=2}^n  T_{zz}(z_i)  \prod_{j=1}^N V_{\alpha_j, \hat g}(x_j) \rangle_{\hat g} dv_{\hat g}(z),
\end{align*}
and stripping out the arbitrary test function $f$ we conclude
\begin{align*}
\langle T_{zz}(z_{1})\prod_{i=2}^{n} T_{zz}(z_{i}) \prod_j &V_{\alpha_{j}, \hat g}(x_{j}) \rangle_{\hat g} =-\tfrac{c}{12}t(z_1)\langle\prod_{i=2}^{n} T_{zz}(z_{i})\prod_j V_{\alpha_{j}, \hat g}(x_{j})\rangle_{\hat g}\\&+\tfrac{c}{12}\sum_{k=2}^{n} \left(\frac{6}{(z_1-z_k)^4}+\frac{2t(z_k)}{(z_1-z_k)^2}+\frac{\partial_zt(z_k)}{z_1-z_k}\right)\langle\prod_{i\neq 1,k}^{n} T_{zz}(z_{i})\prod_j V_{\alpha_{j}, \hat g}(x_{j})\rangle_{\hat g}\\&
+\sum_{k=2}^{n}  \left(\frac{2}{(z_1-z_k)^2}+\frac{1}{z_1-z_k}\partial_{z_k}\right)\langle \prod_{i=2}^{n} T_{zz}(z_{i})\prod_j V_{\alpha_{j}, \hat g}(x_{j})\rangle_{\hat g}\\
&+\sum_{l} \left(  \frac{\Delta_{\alpha_l}}{(z_1-x_l)^2}+ \frac{\Delta_{\alpha_l}\partial_z\sigma(x_l)}{z_1-x_l}+\frac{1}{z_1-x_l}\partial_{x_l}\right)
\langle\prod_{i=2}^{n} T_{zz}(z_{i})\prod_{j} V_{\alpha_{j}, \hat g}(x_{j})\rangle_{\hat g}.
\end{align*}
This is equivalent with  \eqref{wardid} once one uses the definition \eqref{newT}. \qed

\begin{remark}\label{Tbarz}
Using the computations from the proofs of Propositions \ref{smoothness} and \ref{ward_prop} it is simple to check that
\begin{align}\label{Tbz}
\langle T_{\bar z \bar z}(z_1) \prod_{k=2}^n T_{z z}(z_k) \prod_{j=n+1}^m T_{\bar z \bar z}(z_j) \prod_{i=1}^N V_{\alpha_i,\hat g}(x_i) \rangle_{\hat g} &= - \tfrac{c}{12}  \bar t(z_1) \langle \prod_{k=2}^n T_{z z}(z_k) \prod_{j=n+1}^m T_{\bar z \bar z}(z_j) \prod_{i=1}^N V_{\alpha_i,\hat g}(x_i) \rangle_{\hat g}, \nonumber \\
\langle \prod_{k=1}^n T_{\bar z \bar z} (z_k) \prod_{i=1}^N V_{\alpha_i,\hat g}(x_i) \rangle_{\hat g} &= \overline{\langle \prod_{k=1}^n T_{ z  z} (z_k) \prod_{i=1}^N V_{\alpha_i,\hat g}(x_i) \rangle_{\hat g}}, 
\end{align}
whenever $z_i\neq z_j$ for $i \neq j$ and $z_i \neq x_j$ for all $i$ and $j$, where $t$ is the function from \eqref{tcorre}.

In similar fashion it is also possible to compute that
\begin{align}\label{Tzbz}
\langle T_{z \bar z}(z_1) \prod_{k=2}^n T_{\mu_k \nu_k}(z_k)\prod_{i=1}^N V_{\alpha_i,\hat g}(x_i) \rangle_{\hat g} &= \tfrac{c}{48}  R_{\hat g}(z_1) \langle \prod_{k=2}^n T_{\mu \nu}(z_k)\prod_{i=1}^N V_{\alpha_i,\hat g}(x_i) \rangle_{\hat g}
\end{align}
whenever $z_i\neq z_j$ for $i \neq j$ and $z_i \neq x_j$ for all $i$ and $j$, but we skip this computation since the identity is not relevant from the point of view of the Virasoro algebra discussed in the next section.

These formulae motivate the focus on the $T_{zz}$ correlations: the $T_{\bar z \bar z}$ correlations are obtained from $T_{zz}$ correlations by complex conjugation, and the other correlations are rather trivial. 

\end{remark}

We finish this section with a formulation of the Ward identity for the case where the vertex operators are replaced by a smooth version as follows.
\begin{definition}\label{Fdef}
Let $F: H^{-s}(\rs) \to \C$ be such that $\langle |F| \rangle_g < \infty$.
\begin{enumerate}
\item We call the support of $F$, denoted by $\opn{supp} F$, the smallest closed set $K$ satisfying $F(X)=0$ whenever $\opn{supp}X \subset K^c$. Here $\opn{supp} X$ denotes the support of $X$, which is the usual notion of the support of a distribution.
\item We say that $F$ is smooth if the  derivatives
\begin{align*}
\tfrac{d}{d \epsilon} F(X+ \epsilon f)|_{\epsilon=0}
\end{align*}
exist for all $X \in H^{-s}(\rs)$ and all $f \in C^\infty_0(\rs)$ and they can be written as
\begin{align*}
\tfrac{d}{d \epsilon} F(X+ \epsilon f)|_{\epsilon=0} = \int_\C   h(z) f(z) \, d^2z\,,
\end{align*}
for all $f$ where $h:\rs \to \C$ is continuous. If $F$ is smooth, we denote $\frac{\delta F}{\delta X(z)} := h(z)$.
\end{enumerate}
\end{definition}
Let $F: H^{-s}(\rs) \to \C$ be smooth and $g$ and $\hat g$ be as in the beginning of proof of Proposition \ref{ward_prop}, except now we assume that the support of $f$ does not intersect the support of $F$ (in the previous setting this assumption corresponded to assuming that support of $f$ does not overlap with the points $z_i$ and $x_i$). By Proposition \ref{weyl_anomaly} we have
\begin{align*}
\langle F(X) \rangle_{g}&=\langle F(X) \rangle_{e^{\varphi}\psi^* \hat g} 
=e^{cA(\varphi,\psi^* \hat g)} \langle F((X- \tfrac{Q}{2} \varphi)\circ\psi) \rangle_{\hat g}
\end{align*}
so that  recalling \eqref{uuepsil} and \eqref{varhiepsil} we get
\begin{align*}
\partial_\epsilon|_0\langle F(X) \rangle_{g}
=\int \langle (- \partial_x u(x) X(x)+\tfrac{Q}{2}(u(x)\partial_x\sigma(x)+\partial_x u(x))F_x(X) \rangle_{\hat g}d^2x + c \partial_\epsilon|_0 A(\varphi,\psi^* \hat g) \langle F (X) \rangle_{\hat g},
\end{align*}
where the last term was computed in \eqref{Beq} and we denoted 
\begin{align}\label{funcder}
F_x=\frac{\delta F}{\delta X(x)}.
\end{align}
Integrating by parts and using \eqref{uuepsil} we get
\begin{align*}
\int \langle - \partial_x u(x) X(x) F_x(X) \rangle_{\hat g} d^2x &= - \tfrac{1}{4\pi} \int \frac{f(z)}{x-z} \langle \partial_x X(x) F_x(X) \rangle_{\hat g} dv_{\hat g}(z) d^2 x\,, \\
\tfrac{Q}{2} \int (u(x) \partial_x \sigma(x) + \partial_x u(x)) \langle F_x(X) \rangle_{\hat g} d^2x &=  \tfrac{Q/2}{4\pi} \int \Big(- \frac{f(z)}{x-z} \partial_x \sigma(x) + \frac{f(z)}{(x-z)^2} \Big) \langle F_x(X) \rangle dv_{\hat g}(z) d^2x.
\end{align*}
We conclude using $\partial_\epsilon|_0 \langle F\rangle_g=:\frac{1}{4\pi} \int f(z) \langle T_{zz}(z) F \rangle_{\hat g} dv_{\hat g}(z)$ that
\begin{align}\label{wardnew}
\langle T(z)F\rangle_{\hat g}=\int\frac{1}{z-x}\langle (\partial_xX+\tfrac{Q}{2}\partial_x\sigma
) F_x\rangle_{\hat g}d^2x+\tfrac{Q}{2}\int\frac{1}{(z-x)^2}\langle F_x\rangle_{\hat g}d^2x,
\end{align}
where $T$ is given by \eqref{newT}. Note that all the terms are well-defined since $F$ is smooth and if $x$ is outside the support of $F$, then $F_x=0$, which especially means that in the above integrals $F_x$ vanishes in a neighbourhood of $z$.

In what follows we are interested in $F$ of the form 
\begin{align}\label{domain}
F(X)=e^{X(h_0)}\prod_{j=1}^N X(h_j)
\end{align}
where $N\geq 0$ and the functions $h_j$ are complex valued and smooth with compact support and $X(h_i)=\int h_i(z)X(z)d^2z$. Furthermore we require $\Re\int_\C h_0>2Q$, where $\Re z$ denotes the real part of $z$, so that $\langle|F|\rangle_{\hat g}<\infty$, as was discussed in the proof of Proposition \ref{correlationf}.  Let $\caV$ denote the linear span of such $F$. Then by the definition of the functional derivative Definition \ref{Fdef} we get $(e^{X(h_0)})_z = h_0(z) e^{X(h_0)}$ and $X(h_0)_z=h_0(z)$. We conclude
\begin{align}\label{F_x}
F_x=(h_0(x)+\sum_{j=1}^N h_j(x)X(h_j)^{-1})F.
\end{align}
Note that $\opn{supp}F$ is the union of the supports of the $h_j$'s. Then, if $z\in (\opn{supp}F)^c$
\begin{align}\label{wardnew1}
\langle T(z)F\rangle_{\hat g}=\langle \caT_zF\rangle_{\hat g}
\end{align}
where $\caT_z F$ is determined by plugging \eqref{F_x} into \eqref{wardnew} and integrating by parts. Explicitly
\begin{align}\label{wardnew2}
(\caT_zF)(X) &:= \Big(X(\tau_z h_0)+\rho_z h_0+\sum_{j=1}^N(X(\tau_z h_j)+\rho_zh_j)X(h_j)^{-1} \Big)F(X) \,, \\
(\tau_z h_j)(x)&:=-\partial_x \Big(h_j(x)\frac{1}{z-x} \Big) \,, \nonumber \\
\rho_z h_j&:= \tfrac{Q}{2}\int \Big(\frac{1}{z-x}\partial_x\sigma(x)+\frac{1}{(z-x)^2} \Big)h_j(x)d^2x\,. \nonumber
\end{align}
It is simple to check that $\caT_zF\in\caV$ and $\opn{supp}(\caT_zF)\subset \opn{supp}F$.

Now to get the Ward identities for the observables \eqref{domain} we proceed as in the proof of Proposition \ref{ward_prop} where the diffeomorphism and Weyl transformation laws from Proposition \ref{tinsert} were used. A direct consequence of Proposition \ref{weyl_anomaly} and the computations in the proof of Proposition \ref{tinsert} is that the transformation laws \eqref{emdiff} and \eqref{emweyl} in this case take the form
\begin{align*}
\langle\prod_{i=1}^n T_{\mu_{i}\nu_{i}}(z_{i}) F\rangle_{\psi^{\ast}g}&=\langle\prod_{i=1}^n \tilde T_{\mu_{i}\nu_{i}}(z_{i}) \psi_* F\rangle_{g} \,,\\
\langle\prod_{i=1}^n T_{\mu_{i}\nu_{i}}(z_{i}) F\rangle_{e^\varphi g} &= e^{cA(\varphi,g)} \langle\prod_{i=1}^n (T_{\mu_{i}\nu_{i}}(z_{i})+a_{\mu_{i}\nu_{i}}(\varphi,g,z_{i})) F(\cdot - \tfrac{Q}{2} \varphi )\rangle_{g}\,.
\end{align*}
Now we get the Ward identity almost identically as in the proof of Proposition \ref{ward_prop} by using the above transformation laws and the computations leading to \eqref{wardnew}. The result is
\begin{align}\nonumber
\langle T(z_{1})\prod_{i=2}^{n} T(z_{i})F\rangle_{g}&=
\sum_{k=2}^{n} \frac{c/2}{(z_1-z_k)^4} \langle\prod_{i\neq 1,k}^{n} T(z_{i})F)\rangle_{g} 
+\sum_{k=2}^{n}  \left(\frac{2}{(z_1-z_k)^2}+\frac{1}{z_1-z_k}\partial_{z_k}\right)\langle \prod_{i=2}^{n} T(z_{i})F\rangle_{g}\nonumber\\
 &+ \langle \prod_{i=2}^n T(z_i) \caT_{z_1} F \rangle_g,
\label{newwward}
\end{align}
whenever $F$ is as in \eqref{domain}, $z_i \neq z_j$ for $i \neq j$ and $z_i \notin \opn{supp} F$ for all $i$.

\section{Prospects: Representation theory}\label{prospects}

The Ward identities have well known algebraic consequences. To formulate these note that by iterating \eqref{newwward}  for $F\in \caV$ we get
\begin{align*}
\langle \prod_{i=1}^{n} T(z_{i})F\rangle_{g}=\langle G\rangle_{g}
\end{align*}
with $G\in \caV$. This can be viewed as an action $T(z):\caV\to\caV$. This action  gives rise to a representation of the {\it Virasoro algebra} on the physical Hilbert space $\caH$ which is canonically related to LCFT. To describe the latter (see \cite{Kup} for details) it is convenient to choose the metric $g=e^\sigma |dz|^2$  with
\begin{align}
\sigma=-2\ln (\bar zz) \mathbf{1}_{|z|\geq 1}
\label{oursigma}
\end{align}
i.e. the metric is Euclidean $ |dz|^2$ on the unit disc $\D$ and $|z|^{-4} |dz|^2$ on  $\D^c$. The curvature of $g$ is concentrated on the equator: $R_g(z) = 4 \delta(|z|-1)$. We denote the  GFF $X_g$ simply by $X$. It has the covariance
\begin{equation}\label{hatGformula11}
\E X(z)X(z')
=\ln\frac{1}{|z-z'|}+
(\ln|z|)\vee 0+(\ln|z'|)\vee 0.
\end{equation}
The function \eqref{rhogamma} becomes in this metric
$$\rho_{\gamma, g}=e^{\sigma}$$
and thus the  chaos measure is 
\begin{align}
dM_{\gamma}(z)= e^{\gamma X(z)-\frac{\gamma^2}{2}\E X(z)^2}dv(z)
\label{oursigma}
\end{align}
where the volume is $dv(z)=e^\sigma d^2z$. 

The Liouville expectation is then given by
\begin{align}\label{FL1} 
 \langle F \rangle
 = \int dc \, e^{ -2Qc}\E [F(X)
  e^{ -\mu e^{\gamma c} 
 M_{\gamma}(\C)}].
\end{align}

Let $\caF_\D$ consist of functions \eqref{domain} with ${\rm supp}\,F\subset\D$ (see Definition \ref{Fdef} for the definition of the support) and $\Re\int_\D f>Q$, where $\Re$ denotes the real part. Let $\theta: \rs \to \rs$ be the reflection $\theta(z) = 1/\bar z$. Given a $F\in \caF_\D$ define \begin{align*}
(\Theta F)(X) &= \overline{F(X \circ \theta)}
\end{align*}
where $X\circ\theta$ refers to the distribution $f\to X(f_\theta)$ where $f_\theta=|z^{-4}|(f\circ \theta)$. Obviously  $\opn{supp}(\Theta F)\subset\D^c$.

We define a sesquilinear form $(\cdot, \cdot): \mathcal{F}_\D \times \mathcal{F}_\D \to \C$ by 
\begin{align}\label{form}
(F,G) &= \langle \overline{(\Theta F)} G \rangle\,.
\end{align}
Since $\Re(\int_\C(f+\bar f_\theta))=2\Re\int_\D f>2Q$ this is well defined. Reflection positivity is the following statement:
\begin{proposition}\label{OS} The form  \eqref{form} is  positive semidefinite:
\begin{align}\label{scalar1}
(F,F)\geq 0.
\end{align}
\end{proposition} 
For proof see  \cite{Kup}. It is simple to check that $(F,G) = \overline{(G,F)}$ by using a decomposition of the GFF described in Proposition 2.2 of \cite{Kup}, but we skip the computation to keep this discussion short.

We define the Hilbert space $\mathcal{H}$ of LCFT as the completion of $\mathcal{F}_\D/\mathcal{N}$, where
\begin{align*}
\mathcal{N} &= \{ F \in \mathcal{F}_\D \mid (F,F) = 0 \}\,.
\end{align*} 
Let $F$ and $G$ have supports in the disc $\D_r$ with $r<1$ and let $C_i$ be circles of radii $1>r_1>r_2\dots >r_k>r$. Define the objects
\begin{align}\label{lnprod}
(F,L_{n_1}\dots L_{n_k}G):= ( \tfrac{1}{2\pi i})^k\oint_{C_1}dz_1 \, z_1^{n_1+1}\dots \oint_{C_k} dz_k\, z_k^{n_k+1}\langle T(z_1)\dots T(z_k) (\Theta F)G \rangle.
\end{align}
Note that the left-hand side is just a notation and as such does not define the operators $L_{n_1} \hdots L_{n_k}: \mathcal{F}_\D \to \mathcal{F}_\D$. Even if we can guess what the definition of the $L_{n_i}$'s is supposed to be (which is readable from the above notation), showing that $L_{n_1} \hdots L_{n_k}$ maps $\mathcal{F}_\D$ into $\mathcal{F}_\D$ seems difficult.

As a consequence of iterating the Ward identities, the integrand on the right-hand side of \eqref{lnprod} is analytic in $z_i\in\D\setminus\D_r$, $z_i\neq z_j$, the right-hand side depends on the contours $C_i$ only through their order, meaning that the value of the integral possibly changes if one swaps $C_i$ and $C_j$ for $i \neq j$. The Ward identities then imply (see \cite{Gaw}) that the commutator $[L_n, L_m] := L_n L_m - L_m L_n$ is given by
\begin{align}\label{virasoroa}
[L_n,L_m]=(n-m)L_{n+m}+\frac{c}{12}(n^3-n)\delta_{n,-m}\,,
\end{align}
where $\delta_{i,j}$ is the Kronecker delta, in the obvious sense as a relation among the objects \eqref{lnprod}. However we would like to realize the $L_n$'s as operators acting on a suitable dense domain in $\caH$ with the adjoints satisfying $L_n^*=L_{-n}$ and construct a representation of the Virasoro algebra \eqref{virasoroa}. This will be the subject of the forthcoming publication \cite{KO}.

\section{Appendix}


We collect here some notations from Riemannian geometry, see for example \cite{Jost}. Let $(\Sigma,g)$ be a smooth compact two-dimensional Riemannian manifold. In this section we use the Einstein summation convention. Given a local coordinate $x=(x^1,x^2)$ we denote $\partial_\alpha=\frac{\partial}{\partial x^\alpha}$. Hence vectors are given as $u=u^\alpha\partial_\alpha$ and covectors as $\lambda=\lambda_\alpha dx^\alpha$. The space of all diffeomorphisms $\psi: \Sigma \to \Sigma$ (smooth maps with smooth inverse) is denoted by $\rm{Diff}(\Sigma)$.

The Riemannian metric $g$ is given by  
\begin{align*}
g(x)=g_{\alpha\beta}(x)dx^\alpha\otimes dx^\beta
\end{align*}
where  $g_{\alpha\beta}(x)$ is a smooth function taking values in positive matrices. The metric $g$ determines a volume measure $v_g$ on $\Sigma$ given in local coordinates by
\begin{align*}
dv_g(x)=\sqrt{\det g(x)}d^2x
\end{align*}
where $d^2x$ is the Lebesgue measure on $\R^2$.
We denote the scalar product by
\begin{align*}
(f,h)_g &= \int \overline{f(x)} h(x) \, dv_g(x),
\end{align*}
and then $L^2(\Sigma,g) := \{ f: \Sigma \to \C \mid (f,f)_g < \infty\}$.

The group of smooth diffeomorphisms ${\rm Diff}(\Sigma)$ acts on the space of smooth metrics by $g\to\psi^\ast g$ where the pullback metric is given in coordinates as
\begin{align}\label{pullback}
(\psi^{\ast}g)(x)=D\psi(x)^{T}g(\psi(x))D\psi(x).
\end{align}
Then we have the change of variables formula
\begin{align*}
(f,h)_g=(f\circ\psi,h\circ\psi)_{\psi^\ast g}.
\end{align*}

We say that two Riemannian metrics $g'$ and $g$ belong to the same conformal class if there exists $\varphi \in C^\infty(\Sigma)$ and $\psi \in \opn{Diff}(\Sigma)$ such that $g'=e^\varphi \psi^* g$. We say that $g'$ and $g$ are conformally equivalent if there exists $\varphi \in C^\infty(\Sigma)$ such that $g' = e^\varphi g$.

Let us denote the inverse of the matrix $g(x)$ by $g^{\alpha\beta}(x)$. The reason for the upper indices is that the tensor field
$ g^{\alpha\beta} \partial_\alpha \otimes \partial_\beta
$ is invariantly defined. 
It allows us to define 
the Dirichlet form
\begin{align}\label{dirichlet_form}
\mathcal{D}_g(f,h) &:=  \int_\Sigma g^{\alpha \beta} \overline{\partial_\alpha f} \partial_\beta h \, dv_g
\end{align}
and 
the Sobolev space $H^1(\Sigma,g)=\{f\in L^2(\Sigma,g):\mathcal{D}_g(f,f)<\infty\}$. The Dirichlet form gives rise to a positive self-adjoint operator $-\Delta_g$   by
\begin{align*}
\mathcal{D}_g(f,h) &= -(f,\Delta_g h)_g\,.
\end{align*}
On smooth functions $f$  by integration by parts one gets the formula for the Laplace--Beltrami operator as
\begin{align*}
\Delta_g f&= \frac{1}{\sqrt{\det g}}  \partial_\alpha( \sqrt{\det g} g^{\alpha \beta}  \partial_\beta f).
\end{align*}
The Dirichlet form satisfies the diffeomorphism invariance
\begin{align*}
\mathcal{D}_g(f,h)=\mathcal{D}_{\psi^*g}(f\circ\psi, h\circ\psi) \,,
\end{align*}
which implies
\begin{align}\label{lb_covariance}
(\Delta_g f)\circ\psi &= \Delta_{\psi^*g} (f\circ\psi)\,.
\end{align}
The Laplace--Beltrami operator has a discrete spectrum $(\lambda_{g,n})_{n=0}^\infty$, $0=\lambda_0 < \lambda_1 \leq \lambda_2 \dots$, and a complete (in $L^2(\Sigma,g)$) set of smooth eigenfunctions $(e_{g,n})_{n=0}^\infty$ with $e_{g,0}$ the constant function. 
The property \eqref{lb_covariance} implies
\begin{align}\label{eigen_covariance}
e_{g,n} \circ\psi&= e_{\psi^* g,n}\,, \nonumber \\
\lambda_{g,n} &= \lambda_{\psi^*g, n}\,.
\end{align}
The zero-mean Green's function is defined by the formula
\begin{align*}
G_g(x,y) &= \sum_{n=1}^\infty \frac{e_{g,n}(x) e_{g,n}(y)}{\lambda_{g,n}}\,.
\end{align*}
Then \eqref{eigen_covariance} implies
\begin{align}\label{green_covariance}
G_{\psi*g}(x,y) &= G_g(\psi(x),\psi(y))\,.
\end{align}

We can view the diffeomorphisms also passively as changes of coordinates. Then locally we can find a coordinate so that $g_{\alpha\beta}=e^\sigma\delta_{\alpha\beta}$ with a smooth $\sigma$ (the proof is a small variation of Proposition \ref{beltrami_solu}). An atlas of such coordinates defines a {\it complex structure} on $\Sigma$ since the transition functions are easily seen to be analytic. Indeed, if on $\R^2$ we have $g=\psi^\ast h$ with $g$ and $h$ diagonal matrices then $\psi(x^1,x^2)=(u(x^1,x^2),v(x^1,x^2))$ where $u,v$ satisfy the Cauchy-Riemann equations. We can then introduce  complex coordinates $z = x^1 + i x^2\,,
\bar z = x^1 - i x^2\,$ and write tensors using them. E.g. $T=T_{\alpha\beta}dx^\alpha \otimes dx^\beta$ becomes
\begin{align}\label{ccoordinates}
T &= T_{zz} dz \otimes dz + T_{\bar z \bar z} d \bar z \otimes d \bar z  + T_{z \bar z} dz \otimes  d \bar z + T_{\bar z z} d\bar z \otimes dz\,,
\end{align}
where
\begin{align*}
T_{zz} &= \tfrac{1}{4} (T_{11} - T_{22} - 2i T_{12}), \\
T_{z \bar z} &=T_{\bar z  z}= \tfrac{1}{4}(T_{11} + T_{22}),
\end{align*}
and $T_{\bar z \bar z} = \overline{T_{zz}}$.
Furthermore, for a $2 \times 2$ symmetric matrix $f$ the formulae for the determinant and the trace in complex coordinates are
\begin{align}\label{det}
\det f &= 4 f_{z \bar z}^2 - 4 f_{zz} f_{\bar z \bar z}\,, \\
\opn{tr} f &= 4 f_{z \bar z}\,.
\end{align}
See for example Section 2.9.1 in \cite{astala} for some other basic properties of the complex coordinates.

We denote the scalar curvature of $g$ by $R_g$. It is defined by contracting the Ricci tensor $R_{\mu \nu}$
\begin{align}\label{scalar curvature}
R_g &:= g^{\mu \nu} R_{\mu \nu}\,,
\end{align}
where the Ricci tensor comes from contracting the Riemann tensor
\begin{align*}
R_{\mu \nu} &= R^{\alpha}_{\mu \alpha \nu}\,.
\end{align*}
Finally, the Riemann tensor is defined by the formula
\begin{align*}
R^{\alpha}_{\beta \gamma \delta} &= \partial_\gamma \Gamma^\alpha_{\beta \delta} - \partial_\delta \Gamma ^\alpha_{\beta \gamma} + \Gamma^\alpha_{\lambda \gamma} \Gamma^{\lambda}_{\beta \delta} - \Gamma^\alpha_{\lambda \delta} \Gamma^{\lambda}_{\beta \gamma}\,,
\end{align*}
where the $\Gamma$'s are the Christoffel symbols
\begin{align*}
\Gamma^\alpha_{\beta \gamma} &= \hf g^{\alpha \delta}(\partial_\beta g_{\delta \gamma} + \partial_\gamma g_{\delta \beta} - \partial_\delta g_{\beta \gamma})\,.
\end{align*}
As $R_g$ is defined by contractions of the metric and its derivatives, under diffeomorphisms it transforms as
\begin{align}\label{curvature_covariance}
\psi^* R_g &= R_{\psi^*g}\,.
\end{align}

In coordinates where $g_{\alpha\beta}= e^\sigma \delta_{\alpha\beta}$
 we have the formula
\begin{align*}
R_g &=  -4e^{-\sigma} \partial_z \partial_{\bar z} \sigma\,.
\end{align*}
From this  together with 
\eqref{curvature_covariance} and the existence of conformal coordinates it is easy to infer that in general
\begin{align}\label{curvature_weyl}
R_{e^\varphi g} &= e^{-\varphi}(R_g - \Delta_g \varphi)\,.
\end{align}

As an application of these definitions we have the Lemma used in the text:

\begin{lemma}\label{variationlemma}
Let 
 \begin{align*}
 F(g)=\int h R_gdv_g.
\end{align*}
Then
 \begin{align*}
 \caF(z):=\frac{\delta}{\delta g^{zz}(z)}F(g)|_{g=e^{\sigma} \delta}=-\partial_z^2h+\partial_z\sigma \partial_zh
\end{align*}
in coordinates where $g_{\alpha\beta}= e^\sigma \delta_{\alpha\beta}$.
\end{lemma}
\proof
By definition
\begin{align*}
\partial_{\epsilon}\big|_{0} F(g_\eps)=\int \caF(z)\phi(z) dv_g(z)
\end{align*}
where $g_\eps^{z\bar z}=2e^{-\sigma}$ and  $g_\eps^{z z}=\eps \phi$. Let $\tilde g_\eps=e^{-\sigma}g_\eps$. Then
$$
R_{ g_\eps}dv_{g_\eps}=(R_{ \tilde g_\eps}-\Delta_{\tilde g_\eps}\sigma)dv_{\tilde g_\eps}
$$
and  $\tilde g_\eps^{z\bar z}=2$ and  $\tilde g_\eps^{z z}=\eps e^\sigma \phi$.  
We have
\begin{align}
\partial_{\epsilon}\big|_{0} R_{\tilde g_\eps}&=-\partial_z^2 (e^\sigma \phi) 
\label{riccci1111}
\end{align}
$\partial_{\epsilon}\big|_{0}v_{\tilde g_\eps}=0$ and 
$$\partial_{\epsilon}\big|_{0}\Delta_{\tilde g_\eps}\sigma=\partial_z(e^\sigma \phi\partial_z\sigma).$$
Hence
\begin{align*}
\partial_{\epsilon}\big|_{0} F(g_\eps)&=\int h (-\partial_z^2 (e^\sigma \phi)-\partial_z(e^\sigma \phi\partial_z\sigma))d^2z\\&=\int (-\partial_z^2h+\partial_z\sigma \partial_zh)\phi dv_g
\end{align*}
 which yields the claim. \qed

\end{document}